%% file: main.tex
\documentclass[runningheads]{llncs}
\usepackage{amsfonts}
\usepackage{amsmath,amssymb,enumerate}
\usepackage{ifthen,latexsym,syntonly}
\usepackage{rotating}
\usepackage{lscape}
\usepackage{graphicx}
\usepackage{color}
\usepackage{todonotes}
\usepackage[hypertexnames=false,hyperfootnotes=false]{hyperref}
\usepackage{cleveref}
\usepackage{url}
\usepackage{xcolor}
\usepackage{placeins}
\usepackage{float}
\usepackage[normalem]{ulem}
\providecommand{\U}[1]{\protect\rule{.1in}{.1in}}

\makeatletter
\renewcommand{\inst}[1]{\ensuremath{^{#1}}}
\makeatother

\usepackage{ifthen}

\provideboolean{submissionversion}

\ifthenelse{\boolean{submissionversion}}{
  \renewcommand{\todo}[2][1]{}
}{
}

\provideboolean{showedits}
\setboolean{showedits}{true}

\ifthenelse{\boolean{showedits}}{
  
  \newcommand{\soutred}[1]{\textcolor{red}{\sout{#1}}}
}{
  
  \newcommand{\soutred}[1]{}
}

\provideboolean{shownewedits}
\setboolean{shownewedits}{true}

\ifthenelse{\boolean{shownewedits}}{
  
  \newcommand{\newsout}[1]{\textcolor{orange}{\sout{#1}}}
}{
  
  \newcommand{\newsout}[1]{}
}

\title{A Framework for Combined Transaction Posting and Pricing for Layer 2 Blockchains}

\author{Shouqiao Wang\inst{1} \and Davide Crapis\inst{2} \and Ciamac C. Moallemi\inst{1}}
\authorrunning{Shouqiao et al.}
\institute{
  \inst{1}Columbia University \\
  \inst{2}Ethereum Foundation
}

\begin{document}
\maketitle

\begin{abstract}
  This paper presents a comprehensive framework for transaction posting and pricing in Layer 2 (L2) blockchain systems, focusing on challenges stemming from fluctuating Layer 1 (L1) gas fees and the congestion issues within L2 networks. Existing methods have focused on the problem of optimal posting strategies to L1 in isolation, without simultaneously considering the L2 fee mechanism.  In contrast, our work offers a unified approach that addresses the complex interplay between transaction queue dynamics, L1 cost variability, and user responses to L2 fees. We contribute by (1) formulating a dynamic model that integrates both posting and pricing strategies, capturing the interplay between L1 gas price fluctuations and L2 queue management, (2) deriving an optimal threshold-based posting policy that guides L2 sequencers in managing transactions based on queue length and current L1 conditions, and (3) establishing theoretical foundations for a dynamic L2 fee mechanism that balances cost recovery with congestion control. We validate our framework through simulations.
\end{abstract}


\input{intro-litrev}

\input{model}

\input{optimal-posting}

\input{L2-pricing}

\input{simulation}

\input{conclusion}

\bibliographystyle{plainurl}
\bibliography{name}

\newpage

\input{appendix}

\end{document}

%% file: intro-litrev.tex
\section{Introduction}

The scalability roadmap for Ethereum is rapidly coming to fruition, marked by the deployment of numerous Layer 2 (L2) blockchains that have significantly enhanced the network's capacity and efficiency. Major L2 solutions like Optimism, Arbitrum, and others have attracted assets worth billions of dollars and are processing millions of transactions daily. These platforms offload transaction processing from the main Ethereum blockchain (L1), enabling higher throughput and lower fees for users while maintaining the security guarantees of the underlying network.

A pivotal development facilitating this growth is Ethereum Improvement Proposal 4844 (EIP-4844), which has unlocked ample data availability for L2 rollups. By introducing a new transaction type that carries ephemeral data (blobs), EIP-4844 allows L2 blockchains to post large amounts of data to L1 at significantly reduced costs. This enhancement has resulted in very low fees for L2 users, further incentivizing adoption and driving transaction volumes.

However, the burgeoning L2 ecosystem faces two critical challenges:

\begin{enumerate} \item \textbf{Fluctuating L1 Fees and Posting Costs:} The cost of posting transactions from L2 to L1 is subject to significant variability due to fluctuating L1 gas prices. These fluctuations impact the operational costs of rollups, making it challenging to predict and manage expenses effectively. This issue is poised to become even more pronounced as data availability demand begins to match supply, potentially increasing the cost and variability of data blobs used in rollup operations.

\item \textbf{Limited Capacity and Congestion Pricing:} Despite their enhanced capacity, L2 rollups have finite resources and may experience congestion during periods of high demand. To manage this, they might need to implement congestion fees, adjusting transaction costs to regulate network usage and maintain performance standards.
\end{enumerate}

Addressing these challenges is essential for the sustainable growth of L2 solutions. Current studies have predominantly focused on the problem of optimal posting strategies to L1 in isolation, without concurrently considering the L2 fee mechanism. Practical systems have employed heuristic approaches to set L2 fees. Notably, Arbitrum Nitro has introduced a decoupled fee structure that separates fees into components aimed at recovering L1 posting costs and managing L2 congestion, drawing inspiration from Ethereum's EIP-1559 mechanism. However, these methods lack rigorous theoretical underpinnings, leaving a gap in understanding the optimal interplay between posting strategies and pricing mechanisms.

In this paper, we develop a comprehensive transaction management framework for L2 blockchain systems that captures the complex interplay between transaction queue dynamics, posting strategies, and gas price fluctuations. For the first time, we jointly consider the problems of posting and pricing, providing a unified approach to optimize both aspects simultaneously. Our model operates in discrete time, aligned with L1 block intervals, and incorporates realistic behaviors of both transaction arrivals influenced by L2 fees and L1 gas price movements.

We explicitly set up the objectives of balancing operational costs and controlling congestion within the L2 network. By formulating an objective function that incorporates both posting costs and queuing delays, we aim to derive strategies that optimize the cumulative operational cost over time while maintaining network performance. Furthermore, we establish the first theoretical results for L2 pricing mechanisms under this framework, subject to technical assumptions that we validate through robust simulations.

Our primary contributions are as follows:

\begin{enumerate} \item \textbf{Joint Modeling of Posting and Pricing Strategies:} We introduce a dynamic model that simultaneously considers the optimal posting of transactions from L2 to L1 and the L2 fee mechanism. This joint approach captures the dependencies between transaction queue dynamics, L1 gas price fluctuations, and user behavior in response to L2 fees.

\item \textbf{Optimal Posting Strategy with Threshold Policy:} We derive an optimal posting strategy for the L2 sequencer using dynamic programming techniques. We prove that a threshold policy is optimal, where the sequencer decides to post all pending transactions or none based on a critical queue length that depends on the current L1 gas price. This finding simplifies the decision-making process and provides practical guidelines for L2 operators.

\item \textbf{Theoretical Foundations for L2 Pricing Mechanisms:} We establish the first theoretical results for L2 pricing, explicitly formulating the objectives of cost recovery and congestion control. By modeling the transaction arrival rate as a function of the L2 fee, we develop a dynamic fee adjustment strategy that ensures budget balance and manages network congestion. We prove the existence and uniqueness of optimal fees and show that with our fee adjustment mechanism, fees converge to the optimal fees, under some technical conditions. We also demonstrate the robustness of our approach in more general and practical settings through simulations.
\end{enumerate}

The remainder of the paper is structured as follows. In Section~2, we detail our transaction management model, including the dynamics of the transaction queue, L1 gas price modeling, cost structure, compensation mechanism, and the formulation of the objective function. Section~3 is dedicated to deriving the optimal posting strategy. We employ dynamic programming techniques to establish the Bellman equation for our system and prove that a threshold policy is optimal for the sequencer's decision-making process. In Section~4, we focus on the design of the L2 fee mechanism. We analyze how the L2 fee influences transaction arrival rates and develop a dynamic fee adjustment strategy that achieves budget balance and congestion control. We establish the first theoretical results for L2 pricing under this framework and present an adaptive algorithm for fee updates. Section~5 provides an analysis of the proposed mechanisms, including the assumptions made, convergence results, and discussions on practical implications. We validate our theoretical findings through simulations, demonstrating the robustness of our approach under various realistic network conditions.


\subsection{Literature Review}


Existing literature on decentralized L1 blockchains such as Bitcoin and Ethereum has extensively examined fee market designs, often using game-theoretic models to investigate how auction-based or monopolistic pricing schemes affect miner incentives, user behavior, and network throughput \cite{basu2019towards,lavi2022redesigning,yaish2023correct,yao2018incentive}. Roughgarden \cite{roughgarden2020transaction} provides an economic analysis of Ethereum’s EIP-1559, highlighting dynamic base-fee adjustments for congestion control. Further refinements in L1 fee mechanisms appear in works by Leonardos et al. \cite{opt_chaos}, who analyzed fee market dynamics and demonstrated that optimal fee mechanisms can be achieved despite inherent market chaos, and Crapis et al. \cite{multidimensional}, who investigated optimal dynamic fees for blockchain resources, introducing models that adjust to network conditions to optimize fee structures. Additionally, Crapis \cite{crapis2023eip4844} has provided an analysis of the fee market under EIP-4844, discussing its implications for L1 fee dynamics. While these studies primarily focus on decentralized settings where miners or validators collectively enforce protocol rules, our work centers on a more centralized L2 context, where a single operator unilaterally defines both posting and pricing policies, thereby departing from the standard L1-centric paradigm. Although congestion and blockspace constraints remain central concerns, L2 rollups face the additional challenge of covering volatile posting costs on L1, which substantially shapes the L2 fee design. By jointly modeling L2 queue dynamic and L1 gas prices dynamics, our framework emphasizes maximizing throughput while achieving budget balance on the L2, thus bridging a key gap in the existing literature.

Studies on optimal posting strategies for L2 rollups have emerged to tackle the challenges posed by unpredictable L1 gas prices and the imperative for timely data finalization. Mamageishvili and Felten \cite{eff_batch_posting} propose a Q-learning approach to determine the optimal moment for a rollup to publish transactions on L1, modeling a trade-off between waiting for favorable gas prices and incurring delay costs under a quadratic delay cost model. Although their work offers valuable insights into the batch posting decision, it does not address the design of L2 fee mechanisms or account for the dynamics of user demand driven by network congestion. In contrast, our work considers user demand dynamics via an L2 fee model while adopting a linear delay cost framework. By rigorously proving a threshold policy property, we implement a policy iteration algorithm that is significantly more efficient than the Q-learning method. Similarly, Bar-On and Mansour \cite{bar2023optimal} build on related ideas by offering threshold-based policies for specific classes of cost functions, thereby providing analytical insights into optimal posting schedules. However, their analysis is confined solely to the posting decision and does not examine how user fees might be dynamically set to influence transaction arrivals. In contrast, our framework jointly addresses both posting and pricing decisions, capturing the intricate interplay between volatile L1 gas costs and user-driven congestion. Moreover, we design and rigorously prove key properties of our dynamic L2 fee mechanism under relaxed assumptions. Furthermore, Crapis et al. \cite{crapis2023eip} investigate the economics of EIP-4844, focusing on blob posting strategies and equilibrium cost-sharing among rollups. While their analysis highlights important trade-offs for rollups operating under blob constraints, our work focuses specifically on batch posting. Additionally, we incorporate a dynamic L2 fee mechanism to actively manage network congestion. By unifying the optimal posting strategy and L2 fee mechanism, our approach enables L2 systems to simultaneously maximize throughput and maintain budget balance, thereby offering a comprehensive and robust solution that extends beyond the scope of the aforementioned works.


%% file: model.tex
\section{Model} \label{model}

In what follows, we describe the main components of our transaction management framework for a Layer 2 blockchain system. This framework not only captures the dynamics of the transaction queue and gas prices but also outlines a decision model for optimizing transaction posting costs and waiting costs over time. We consider a realistic transaction processing scenario where both the demand for transaction posting and fluctuating gas prices are modeled in discrete time, indexed by $t$.

\medskip
\noindent\textbf{Queue Dynamics.} The L2 queue length at time $t$, denoted by $Q_{t}$, evolves according to the following equation:
\[
Q_{t+1} = Q_{t} - S_{t} + A_{t},
\]
where $S_{t}$ is the number of transactions posted in the $t_\text{th}$ L1 block, and $A_{t}$ represents the number of incoming transactions between the $t_\text{th}$ and $(t+1)_\text{th}$ L1 block time. While $A_{t}$ is modeled as an independent and identically distributed process conditioned on a specific Layer 2 gas price $g$, the distribution of $A_{t}$ can vary as $g$ changes, which captures the dynamic interplay between the transactions processed and new ones arriving.

\medskip
\noindent\textbf{L1 Gas Price Dynamics.} The L1 gas price $P_{t}$, which is exogenous to the L2 structures, follows a mean-reverting process, i.e., it tends to move back toward a long-term average over time, with dynamics given by
\[
P_{t+1} = \theta \mu + (1 - \theta) P_{t} + \sigma \omega_{t},
\]
where $\theta$ controls the rate of mean reversion towards the long-term average $\mu$, $\sigma$ is the volatility parameter, and $\omega_{t}$ follows an i.i.d.\ standard normal distribution. 
This is also known as an autoregressive AR(1) process.

The mean-reverting behavior of $P_{t}$ is particularly realistic in the context of Ethereum's EIP-1559 update, which introduces a mechanism for adjusting transaction fees that inherently aims to  stabilize block sizes around a target size. Specifically, the gas fee for the next block increases if the current demand exceeds the target, and decreases if it falls below. If current demand exceeds the target then prices will increase, but as the demand subsequently responds to the increased prices, it will decrease, and in turn drive down future prices. This inherently leads to a mean-reverting dynamic. See \cite{roughgarden2020transaction} for an economic analysis of EIP-1559. Using a mean-reverting process to model $P_{t}$ captures the realistic features of EIP-1559 while simplifying the model for better analysis. This approach is also supported in the literature; for example, \cite{meister2024gas} employ a fractional Ornstein-Uhlenbeck  process to model gas fees, which is also a type of mean-reverting process.

\medskip
\noindent\textbf{Cost Structure.} The cost incurred at the $t_\text{th}$ L1 block is given by:
\[
c(Q_{t}, S_{t}, P_{t}) = a(Q_{t} - S_{t}) + (b_0 + b_1 S_{t}) P_{t} \mathbf{1}_{\{S_{t} > 0\}},
\]
where $a$ is the proportional penalty on the number of waiting transactions after posting, $b_0$ is the fixed cost when transactions are posted, and $b_1$ scales linearly with the number of posted transactions, reflecting the sensitivity of the cost to the number of transactions. Observe that charging a linear penalty on the number of waiting transactions in each time step is equivalent to charging a linear penalty in the total number of time steps each transaction waits \cite{little2008little}. To see this more clearly, note that at each time step, every waiting transaction incurs a cost of $a$. If a transaction waits over multiple time steps, it thus accumulates the sum of these step-by-step penalties. Interchanging the summation over time and over transactions (an argument analogous to summation by parts) demonstrates that this is precisely the same as penalizing the total waiting time for each transaction. Hence, the first term in the cost structure is fundamentally a penalty on queuing delay.

\medskip
\noindent\textbf{Compensation Mechanism.} To mitigate user dissatisfaction due to transaction delays and to align the incentives of the sequencer, the L2 platform could implement a compensation mechanism. For every L1 block of delay experienced by a transaction, the platform refunds the user an amount proportional to the delay, calculated as $a$ times the number of delayed L1 blocks.

This compensation ensures that from the user's perspective, the utility remains consistent irrespective of the transaction posting time. Simultaneously, this approach aligns the incentives of the L2 sequencer with the cost function $c(Q_{t}, S_{t}, P_{t})$ they are optimizing. By providing such refunds, the L2 sequencer demonstrates a commitment to minimizing transaction delays, rather than exploiting these for potential profit. This transparent compensation also reinforces trust in the platform's operational integrity.

\medskip
\noindent\textbf{Objective Function.} The overall objective is to minimize the expected cumulative discounted cost, given by:
\begin{equation*}
J(Q_{t}, P_{t}) = \min_{\{S_{s}\}}
\mathbb{E}\left[\left. \sum_{s \geq t} \gamma^{s-t} \left(a(Q_{s} - S_{s}) + (b_0 + b_1 S_{s}) P_{s} \mathbf{1}_{\{S_{s} > 0\}} \right) \right| Q_t, P_t\right],
\end{equation*}
where $\gamma$ is the discount factor that captures the present value of future costs. The aim is to strike a balance between reducing the delay cost and managing posting costs effectively over time.

%% file: optimal-posting.tex
\section{Optimal Posting Strategy}

We continue our analysis of the optimal posting strategy under the assumption that the arrivals 
$A_t$ are i.i.d., given a fixed Layer 2 fee. This section builds upon the established objective function and focuses on optimizing the transaction posting mechanism under dynamic gas price conditions. The decision variable, \(S_t\), determines the number of transactions to post per L1 block. It is optimized by evaluating immediate costs and forecasting future system states to enhance transaction processing efficiency on Layer 2.

\begin{theorem}
\label{threshold-policy}
\textbf{Threshold Policy.} For any state $(Q_{t}, P_{t})$, the optimal posting strategy $S_{t}^*$ satisfies:
\[
S_{t}^* = 
\begin{cases} 
Q_{t} & \text{if } Q_t > Q^*(P_t), \\
0 & \text{if } Q_t \leq Q^*(P_t),
\end{cases}
\]
where $Q^*(P_t)$ is a critical threshold dependent on the gas price $P_t$.
\end{theorem}

\begin{proof} 
    The proof is provided in Appendix \ref{proofs}.    \qed
\end{proof}

The threshold policy mandates a binary decision: either clear the queue by posting all transactions or withhold all transactions when conditions do not favor posting, based on the queue length relative to a dynamic threshold $Q^*(P_t)$. This straightforward, binary approach greatly simplifies operational decisions by removing intermediate options, which streamlines the posting mechanism and enhances system efficiency by ensuring that decisions are consistently aligned with current economic conditions.

By leveraging the threshold structure of the optimal policy, we significantly simplify the decision-making process in our Markov Decision Problem \footnote{A Markov Decision Problem models sequential decision-making, where a system transitions between states based on chosen actions, with the goal of minimizing expected costs or maximizing rewards.} \cite{puterman2014markov}. This structure transforms the action space from $\{0,1,\cdots,Q\}$ to just two discrete options $0$ or $Q$. Utilizing this binary action space, we implemented an efficient policy iteration algorithm. At each iteration, we evaluate the expected total cost for both possible actions by calculating the immediate cost and the expected future cost based on the current value function estimate. By directly comparing these two options, we can promptly update the policy without exhaustively searching over all possible actions. In contrast, \cite{eff_batch_posting} employs a Q-learning algorithm for the optimal posting strategy without utilizing a threshold policy, resulting in a computation time of approximately 72 hours for a single run, whereas our method completes in just about 6 seconds.

%% file: L2-pricing.tex
\section{Layer 2 Pricing}

In our model, we assume a uniform L2 fee per transaction, simplifying the varying costs typically seen due to different gas usages per transaction in reality. The fee mechanism in our analysis, similar to many L2 platforms, fundamentally addresses two critical objectives: maintaining budget balance and managing congestion. Therefore, our strategy seeks to optimize network throughput within the framework established by these baseline conditions, ensuring that enhancements in throughput do not compromise our commitment to budget balance and effective congestion management.

The objectives of the fee mechanism in our study, and similarly in many L2 platforms, are twofold: to maintain budget balance and to manage congestion effectively. A similar idea is also mentioned in \cite{bousfield2022arbitrum}. The goal is to satisfy the minimum level of these two objectives; under this condition, we maximize the chain throughput. Our mechanism is well-suited for implementation on most optimistic rollup platforms, such as Arbitrum, Base, and Optimism.

\medskip
\noindent\textbf{Arrival Rate.} We denote the arrival rate of transactions by \( \lambda(g) \), which represents the expected number of incoming transactions between L1 block times, \( \mathbb{E}[A(t;g)] \), where \( A(t;g) \) represents the random variable of the number of incoming transactions between L1 block times, whose distribution varies as a function of the fee \( g \). In our model, we assume that \( A(t;g) \) follows a Poisson distribution \cite{katti1968handbook}, which is commonly used to describe event counts that have independent increments and occur with a constant rate. We further assume a linear relationship between the arrival rate and the fee:
$$\lambda(g) = \lambda_0 - kg,$$
where \( \lambda_0 \) is the maximum potential arrival rate when the fee is zero, and \( k \) is a constant that captures the sensitivity of the arrival rate to changes in the fee. If the fee exceeds the threshold \( g > \lambda_0/k \), the arrival rate becomes zero, i.e., no arrivals occur. The linear demand curve assumption is supported by empirical analysis of \cite{pantera2024pricing}. While this assumption simplifies the analysis of the root existence condition, our fee mechanism is robust and applicable to a broad range of arrival rate models.

\medskip
\noindent\textbf{Budget Balance.} Budget balance is achieved when the total fees collected from transactions match the total operational costs of processing those transactions within the L2 framework. To establish this equilibrium, we define the fee \( f \) to satisfy the following condition under the stationary distribution:
\begin{equation}
\mathbb{E} \left[ A(t;f) f - c(S_t^*, Q_t, P_t; f) \right] = 0,
\label{budget-balance}
\end{equation}
where \( c(S_t^*, Q_t, P_t; f) \) is the cost associated with the optimal posting strategies given the L2 fee \( f \). The root of this equation, \( f = f^* \), defines the fee level at which the network achieves budget balance, ensuring financial sustainability by perfectly aligning revenues with costs.

\begin{property} \label{monotonic-posting}
    $\mathbb{E}\left[ c(S_t^*,Q_t,P_t;f) \right]$ is strictly monotonically decreasing with respect to $f$.
\end{property}
\begin{proof}
    The proof is provided in Appendix \ref{proofs}.  \qed
\end{proof}

\begin{corollary}
    The expected profit and loss per L1 block, given by
    \[
    \mathbb{E}\left[ A(s; f) f - c(S_s^*, Q_s, P_s; f) \right],
    \]
    is strictly monotonically increasing for \( f \in [0, \lambda_0/(2k)] \). 
\end{corollary}
\begin{proof}
The corollary is straightforward, since
\begin{align*}
    \mathbb{E}\left[ A(s; f) f - c(S_s^*, Q_s, P_s; f) \right] &= \lambda(f) f - \mathbb{E}[- c(S_s^*, Q_s, P_s; f)] \\
    &= (\lambda_0-kf)f-\mathbb{E}[c(S_s^*, Q_s, P_s; f)],
\end{align*}
where the product $(\lambda_0-kf)f$ is strictly monotonically increasing for $f \in [0, \lambda_0/(2k)] $, and the expected cost $\mathbb{E}[c(S_s^*, Q_s, P_s; f)]$ is strictly monotonically decreasing due to Property~\ref{monotonic-posting}.  \qed
\end{proof}

\begin{property}
The arrival rate $\lambda(f)=0$ if and only if $\mathbb{E}\left[ c(S_t^*,Q_t,P_t;f) \right]=0$.
\end{property}

\begin{property} \label{cost-bound}
    The cost function $\mathbb{E}\left[ c(S_t^*,Q_t,P_t;f) \right]$ can be bounded by a decreasing linear function. Specifically,
    $$\mathbb{E}\left[ c(S_t^*,Q_t,P_t;f) \right] \leq (b_0+b_1 \lambda_0-b_1 kf) \mu.$$
\end{property}
\begin{proof}
    The proof is provided in Appendix \ref{proofs}.  \qed
\end{proof}

\begin{theorem} \label{existence-bb}
    \textbf{Existence of Unique Budget Balance Fee.} If
    \[
    \frac{\lambda_0^2}{4k} \geq (b_0+b_1 \lambda_0) \mu,
    \]
    then there exists a unique fee $f^* \in [0,\lambda_0/(2k)]$ that achieves budget balance, being the root of the equation defined in equation (\ref{budget-balance}).
\end{theorem}
\begin{proof}
    The proof is provided in Appendix \ref{proofs}.  \qed
\end{proof}

The condition outlined in the theorem ensures that there is a fee level, \( f^* \), which precisely balances the revenues from transaction fees with the costs of transaction processing, thereby achieving budget balance. Setting the fee below \( f^* \) results in revenues that fail to cover the operational costs, leading to financial losses for the platform. Conversely, setting the fee above \( f^* \) may generate surplus revenue, potentially turning a profit.

\medskip
\noindent\textbf{Target Arrival Rate.} In managing the L2 network, it is critical to recognize the system's maximum capacity for executing transactions. The target arrival rate, denoted as \( \bar{\lambda} \), is set based on this maximum capacity to ensure the network operates efficiently without being overwhelmed by an excessive volume of transactions.

\medskip
\noindent\textbf{Congestion Control.} To align the actual arrival rate of transactions with the target \( \bar{\lambda} \), a congestion control fee \( p \) is utilized. The fee that precisely balances the incoming transaction rate with the network’s capacity is \( p^* \), defined by the root of the equation:
\begin{equation}
    \mathbb{E}[\bar{\lambda} - A(t; p)] = 0.
    \label{congestion-control}
\end{equation}
The fee \( p^* \) ensures that the number of transaction arrivals matches the target arrival rate \( \bar{\lambda} \).

\begin{theorem} \label{existence-cc}
    \textbf{Existence of Unique Congestion Control Fee.} If
    \[
    \frac{\lambda_0}{2} \leq \bar{\lambda} \leq \lambda_0,
    \]
    then there exists a unique fee $p^* \in [0,\lambda_0/(2k)]$ that achieves congestion control, being the root of the equation defined in equation (\ref{congestion-control}).
\end{theorem}
\begin{proof}
    This proof is straightforward, because we only need to solve
    $$0 = \mathbb{E}[\bar{\lambda} - A(t; p^*)] = \bar{\lambda} - \lambda(p^*) = \bar{\lambda} - \lambda_0 + kp^*.$$
    The equation has a unique root $p^* = (\lambda_0-\bar{\lambda}) /k.$
    The value $p^* \in [0,\lambda_0/(2k)]$ if and only if $\lambda_0/2 \leq \bar{\lambda} \leq \lambda_0$.  \qed
\end{proof}

The condition outlined in the theorem ensures that there is a fee level, \( p^* \), which precisely aligns the actual transaction arrival rate with the target rate \( \bar{\lambda} \), effectively managing congestion. Setting the fee below \( p^* \) may lead to network overload, while setting it above \( p^* \) can result in underutilization of network capacity. 

\medskip
\noindent\textbf{Optimal Fee Strategy.} Our fee mechanism is strategically designed to simultaneously achieve budget balance and manage congestion, fundamental conditions for the stable operation of an L2 platform. To this end, the fee must be set at least as high as \( f^* \) to cover operational costs and ensure financial sustainability, and at least as high as \( p^* \) to regulate the flow of transactions and prevent system overload. 

Given our goal to maximize network throughput, which ideally involves keeping fees as low as possible, the optimal fee charged is \( \max(f^*, p^*) \). This strategy ensures that fees are not set higher than necessary to meet the foundational requirements, allowing the platform to process the maximum number of transactions without compromising financial viability or operational stability. By charging \( \max(f^*, p^*) \), we maintain a balance that supports the highest possible throughput within the constraints of budget balance and congestion management.

\subsection{L2 Fee Mechanism} \label{practical-mechanism}

In this part, we introduce a dynamic fee mechanism designed to achieve both budget balance and congestion control while maximizing throughput. When considering only budget balance, we have a fee mechanism based on updating the fee $f$. Similarly, when considering only congestion control, we have a fee mechanism based on updating the fee $p$. Since we may not know in advance whether the budget balance fee $f^*$ is greater than the congestion control fee $p^*$, we propose an adaptive approach that updates both fees based on observed network conditions to determine the optimal fee.

\medskip
\noindent\textbf{Fee Update Mechanism for Budget Balance.} For the budget-balancing fee $f$, we define its update rule as:
\[
f_{t+1} = \Pi_{[0, \lambda_0/(2k)]} \left( f_t - a \cdot X(t; f_t) \right),
\]
where $\Pi_{[x_a, x_b]}(x)$ projects $x$ onto the interval $[x_a, x_b]$ to ensure the fee remains within feasible bounds, and $a > 0$ is a step size parameter that controls the magnitude of each fee update, i.e., larger values lead to more aggressive adjustments, while smaller values yield more conservative changes. Here, $X(t; f_t)$ represents the cumulative profit or loss over the $t$-th posting period when using fee $f_t$, defined as:
\[
X(t; f_t) = \sum_s \left( A(s; f_t) f_t - c(S_s^*, Q_s, P_s; f_t) \right).
\]
In this expression, $A(s; f_t)$ is the observed number of arrivals during the $s$-th L1 block time after applying fee $f_t$, and $c(S_s^*, Q_s, P_s; f_t)$ is the cost associated with the optimal posting strategy at the $s$-th L1 block given fee $f_t$. The summation over $s$ aggregates over multiple L1 block intervals within the $t$-th posting period.
Specifically, suppose that after the $(t-1)$-th posting, the queue is empty, thanks to Theorem \ref{threshold-policy}, the Threshold Policy. During the $t$-th posting period, transactions arrive over several L1 block intervals. We may decide not to post immediately, accumulating transactions and observing the profit and loss in each interval (the revenue from fees minus the cost). We sum these values until we decide to post, obtaining $X(t; f_t)$.

The update rule aims to achieve financial sustainability by aligning fees with operational costs. If $X(t; f_t) < 0$, it suggests the current fee is not covering costs, prompting an increase in $f_{t+1}$. Conversely, if $X(t; f_t) > 0$, there is room to reduce the fee in $f_{t+1}$ without undermining financial health, potentially boosting transaction volume.

\medskip
\noindent\textbf{Fee Update Mechanism for Congestion Control.} For the congestion control fee $p$, we define its update rule as:
\[
p_{t+1} = \Pi_{[0, \lambda_0/(2k)]} \left( p_t - b \cdot Y(t; p_t) \right),
\]
where $b > 0$ is a step size parameter. Here, $Y(t; p_t)$ represents the cumulative difference between the target arrival rate and the actual arrivals over the $t$-th posting period when using fee $p_t$, defined as:
\[
Y(t; p_t) = \sum_s \left( \bar{\lambda} - A(s; p_t) \right).
\]
In this expression, $A(s; p_t)$ is the observed number of arrivals during the $s$-th L1 block time after applying fee $p_t$, and $\bar{\lambda}$ is the target arrival rate that the network aims to maintain to avoid congestion. The summation over $s$ captures the total effect over multiple L1 block intervals within the $t$-th posting period.

The congestion control fee is adjusted to align the actual transaction flow with the target arrival rate. A negative $Y(t; p_t)$ indicates congestion, suggesting a need to increase $p_{t+1}$ to reduce the incoming transaction rates. Conversely, a positive $Y(t; p_t)$ suggests the capacity to handle more transactions, allowing a fee reduction.

\medskip
\noindent\textbf{Adaptive Fee Selection Mechanism.} Since we do not know whether $f^* \leq p^*$ or $f^* > p^*$, we adopt an adaptive mechanism that updates both fees based on observed network performance to determine the optimal fee.

Let $g_t$ represent the fee applied during the $t$-th posting period, and let $\delta_t$ be an indicator variable, representing whether $g_t$ is updated based on the fee update mechanism for budget balance or congestion control. If $g_t$ is updated based on the budget balance fee update mechanism, the indicator variable $\delta_t = 1$; otherwise, the indicator variable $\delta_t = 0$. 


Define \( \zeta(t) = \max \{s \leq t: \delta_s = 1 \}\) and \( \eta(t) = \max \{s \leq t: \delta_s = 0\}\), which represent the last time we have updated the budget balance fee and congestion control fee up to time $t$ respectively. The decision variable $\delta_{t+1}$ and the fee $g_{t+1}$ for the next posting period are updated according to the following rules:

\begin{itemize}
    \item \textbf{If $\delta_t = 1$}:
    \begin{itemize}
        \item If $Y(t; g_t) < 0$, indicating network congestion, we set $\delta_{t+1} = 0$ to switch to the congestion control fee sequence. The corresponding fee update rule
        $$g_{t+1} = \Pi_{[0, \lambda_0/(2k)]} \left( g_{\eta(t)} - b \cdot Y(\eta(t); g_{\eta(t)}) \right).$$
        \item If $Y(t; g_t) \geq 0$, we keep $\delta_{t+1} = 1$, continuing with the budget balance fee sequence. The corresponding fee update rule
        $$g_{t+1} = \Pi_{[0, \lambda_0/(2k)]} \left( g_{\zeta(t)} - a \cdot X(\zeta(t); g_{\zeta(t)}) \right).$$
    \end{itemize}
    \item \textbf{If $\delta_t = 0$}:
    \begin{itemize}
        \item If $X(t; g_t) < 0$, indicating that revenues fail to cover operational costs, we set $\delta_{t+1} = 1$. This decision switches back to the budget balance fee update mechanism for the next period
        $$g_{t+1} = \Pi_{[0, \lambda_0/(2k)]} \left( g_{\zeta(t)} - a \cdot X(\zeta(t); g_{\zeta(t)}) \right).$$
        \item If $X(t; g_t) \geq 0$, indicating financial stability, we keep $\delta_{t+1} = 0$, continuing with the congestion control fee sequence. The corresponding fee update for the next period
        $$g_{t+1} = \Pi_{[0, \lambda_0/(2k)]} \left( g_{\eta(t)} - b \cdot Y(\eta(t); g_{\eta(t)}) \right).$$
    \end{itemize}
\end{itemize}

This method shares similarity with the multi-armed bandit problem, where we iteratively select the fee update mechanism between different objectives based on observed performance. It enables the system to adaptively switch between fee mechanisms as necessary, optimizing for either budget balance or congestion control in response to real-time conditions. This mechanism is particularly robust, adept at handling non-stationary network conditions and ensuring that fees stay close to the ideal level despite fluctuations. Through this adaptive approach, the system maintains crucial controls over budget and congestion, thereby enhancing throughput and ensuring consistent operational efficiency.

\subsection{Analysis}
We analyze the theoretical results for the L2 fee mechanism, considering some relaxations such as decreasing step size and i.i.d. Layer 1 gas fee prices. This analysis aims to establish a robust understanding of the dynamics underpinning the fee update mechanism within a theoretical framework.

\begin{proposition} \label{renewal-thm}
    For any feasible fee \( g \), where \( 0 \leq g \leq \lambda_0/(2k) \), the expected value of \( X(t; g) \) under the stationary distribution can be expressed as:
    \[
    \mathbb{E}[X(t; g)] = \mathbb{E}[\tau(g)] \cdot \mathbb{E}[A(s; g) g - c(S_s^*, Q_s, P_s; g)],
    \]
    where \( \tau(g) \) is the number of L1 blocks between two consecutive postings, which empty the queue according to Theorem \ref{threshold-policy}. Therefore, if \( f^* \) is the fee level such that \(\mathbb{E}[A(f^*) f^* - c(S_s^*, Q_s, P_s; f^*)] = 0\), then \(\mathbb{E}[X(t; f^*)] = 0\). Similarly, for the congestion control metric:
    \[
    \mathbb{E}[Y(t; g)] = \mathbb{E}[\tau(g)] \cdot \mathbb{E}[\bar{\lambda} - A(s; g)].
    \]
    Consequently, if \( p^* \) is the fee level where \(\mathbb{E}[\bar{\lambda} - A(s; p^*)] = 0\), then \(\mathbb{E}[Y(t; p^*)] = 0\).
\end{proposition}

\begin{proof} 
    The proof is provided in Appendix \ref{proofs}.  \qed
\end{proof}

To underpin a well-founded theoretical result, we incorporate the following considerations:
\begin{enumerate}
    \item \textbf{Decreasing Step Size.} The step size parameters \( a_t \) and \( b_t \) decrease over time, defined as \( a_t = a/(t+1) \) and \( b_t = b/(t+1) \).
    \item \textbf{Independent and Identically Distributed L1 Gas Fee.} We assume that the L1 gas fee sequence \( \{p_s\} \) follows an independent and identically distributed distribution.
    \item \textbf{Multiple Observations.} For each fee update, we consider \( \kappa \) postings and use the observations of all these postings to select the fee update mechanism.
\end{enumerate}
Given these considerations, our fee update rules can be described as follows:
\begin{itemize}
    \item If $\delta_t = 1$ and $\sum_{s=\kappa t+1}^{s=\kappa(t+1)} Y(s; g_t) \geq 0$, or if $\delta_t = 0$ and $\sum_{s=\kappa t+1}^{s=\kappa(t+1)} X(s; g_t) < 0$, we set $\delta_{t+1} = 1$, and the fee update rule
    $$g_{t+1} = \Pi_{[0, \lambda_0/(2k)]} \left( g_{\zeta(t)} - a_{i(t)} \cdot \frac{\sum_{s=\kappa \zeta(t)+1}^{s=\kappa (\zeta(t)+1)} X(s; g_{\zeta(t)})}{\kappa} \right),$$
    where \( i(t) = \sum_{s=1}^t \mathbf{1}\{\delta_s = 1\} \), represents the total number of budget balance fee updates up to \( t \).
    \item If $\delta_t = 1$ and $\sum_{s=\kappa t+1}^{s=\kappa(t+1)} Y(s; g_t) < 0$, or if $\delta_t = 0$ and $\sum_{s=\kappa t+1}^{s=\kappa(t+1)} X(s; g_t) \geq 0$, we set $\delta_{t+1} = 0$, and the fee update rule
    $$g_{t+1} = \Pi_{[0, \lambda_0/(2k)]} \left( g_{\eta(t)} - b_{j(t)} \cdot \frac{\sum_{s=\kappa \eta(t)+1}^{s=\kappa (\eta(t)+1)} Y(s; g_{\eta(t)})}{\kappa} \right),$$
    where \( j(t) = t - i(t) \), represents the total number of congestion control fee updates up to \( t \).
\end{itemize}

\begin{theorem} \label{fee-convergence-thm}
    Assuming the conditions and relaxations defined above hold, consider the dynamic fee update mechanisms with \( \delta_t \) indicating whether the fee update is for budget balance or congestion control. Define:
    \begin{itemize}
        \item \( i(t) = \sum_{s=1}^t \mathbf{1}\{\delta_s = 1\} \), the total number of budget balance fee updates up to \( t \),
        \item \( j(t) = t - i(t) \), the total number of congestion control fee updates up to \( t \),
        \item \( f_{i(t)} = g_{\zeta(t)} \), linking the budget balance fee to the last update up to \( t \),
        \item \( p_{j(t)} = g_{\eta(t)} \), linking the congestion control fee to the last update up to \( t \).
    \end{itemize}
    Suppose the condition $\lambda_0^2/(4k) \geq (b_0 + b_1 \lambda_0) \mu$ for the existence of unique budget balance fee in Theorem \ref{existence-bb}, and the condition $\lambda_0/2 \leq \bar{\lambda} \leq \lambda_0$ for the existence of unique congestion control fee in Theorem \ref{existence-cc} hold. Then, the following properties hold almost surely as \( t \rightarrow \infty \):
    \begin{enumerate}
        \item Both \( i(t) \rightarrow \infty \) and \( j(t) \rightarrow \infty \).
        \item The sequences \( f_t \rightarrow f^* \) and \( p_t \rightarrow p^* \).
        \item The long-run average proportions of updates,
        \[
        \frac{i(t)}{t} \rightarrow \pi_f \quad \text{and} \quad \frac{j(t)}{t} \rightarrow \pi_p,
        \]
        where $(\pi_f, \pi_p)$ is the stationary distributions under the transition matrix \( P \). Here, \( P \) governs the transition of \(\delta_t\) when the fee is fixed at the optimal levels, i.e., \( g_t \equiv f^* \) for \( \delta_t^* = 1 \) and \( g_t \equiv p^* \) for \( \delta_t^* = 0 \).
        \item If $f^* \neq p^*$, as the number of observations $\kappa \rightarrow \infty$, the stationary distribution 
        $$(\pi_f, \pi_p) \rightarrow (\mathbf{1}\{f^* > p^*\}, \mathbf{1}\{f^* < p^*\}).$$ 
        As $\kappa \rightarrow \infty$, the convergence rate 
        $$|\pi_f-\mathbf{1}\{f^* > p^*\}| = |\pi_p-\mathbf{1}\{f^* < p^*\}| = O(1/\sqrt{\kappa}). $$
    \end{enumerate}
\end{theorem}

\begin{proof}
    The proof is provided in Appendix \ref{proofs}.  \qed
\end{proof}

According to Theorem \ref{fee-convergence-thm}, the sensitivity of our adaptive fee selection mechanism to network fluctuations can be moderated by adjusting the parameter $\kappa$, which determines the number of postings observed before a fee update. By increasing $\kappa$, the mechanism aggregates more data across multiple L1 block intervals before reconsidering fee changes. This extended observation window tends to smooth out short-term volatility and reduces the frequency of switching between fee update rules, thereby decreasing the likelihood of erratic fee adjustments driven by transient network effects.

In practical terms, the fee mechanism described in Section \ref{practical-mechanism} sets $\kappa=1$ because real-world systems are inherently non-stationary. Updating the fee for each posting, similar to how EIP-1559 updates fees for each block, allows the mechanism to quickly adapt to changing network conditions and improve robustness.
Updating fees after observing a larger set of postings (i.e., increasing $\kappa$) presents a trade-off. On one hand, it improves stability by reducing sensitivity to short-term fluctuations, leading to more consistent fee decisions. On the other hand, it may diminish the mechanism’s responsiveness to rapid changes in network conditions. This reduced agility may compromise robustness, as fees might not adapt swiftly to evolving demand patterns during extended observation periods.

%% file: simulation.tex
\section{Simulation}

In this section, we empirically evaluate our L2 fee mechanism through a series of simulations designed to test its performance under varying conditions. Initially corroborating the theoretical model, we progressively introduce more realistic scenarios shifting from i.i.d., which assumes the L1 gas fees are independent and identically distributed across time, providing a simplified theoretical setting, to non-i.i.d., which models the fees via a mean-reverting process to capture temporal dependence and better reflect real-world dynamics, price distributions and from decreasing to constant step sizes. These simulations aim to assess whether the fees \( f_n \) and \( p_n \) stay close to their theoretical optimal values \( f^* \) and \( p^* \) under less idealized conditions, providing insight into the mechanism's robustness and practical applicability.

\subsection{Simulation Setup}

We conduct simulations to evaluate the performance of our L2 fee mechanism across four distinct scenarios, which vary by the nature of L1 gas fee conditions, either i.i.d. or non-i.i.d., and the approach of step sizes, either decreasing or constant. For each scenario, we assess how well the fee mechanism approaches the theoretical optimal fees \( f^* \) and \( p^* \), as well as analyze the long-term frequency of selecting either the budget balance or congestion control fee update mechanisms. These simulations are crucial for understanding how well the mechanism can maintain its efficiency and effectiveness in different market conditions.

\medskip
\noindent \textbf{Price Generation.} For i.i.d. scenarios, prices are generated from a predefined normal distribution. For non-i.i.d. scenarios, prices are generated using the mean-reverting AR(1) process, as defined in Section \ref{model}.

\medskip
\noindent
\textbf{Step Size Configuration:} For decreasing step sizes, we use \( a_t = a / t \) and \( b_t = b / t \), where \( a \) and \( b \) are predefined. For constant step sizes, \( a_t \) and \( b_t \) are fixed throughout the simulation.

\medskip
\noindent \textbf{Parameters Setup.} Our parameters are calibrated as follows:
\begin{itemize}
    \item \textbf{L1 gas fee parameter.} We use the L1 base fee data during the second half of February 2024 to get the mean and variance. Then for the i.i.d. case, we set its mean $\mu = 3.86 \times 10^{-8}$ and standard deviation $\upsilon = 1.93 \times 10^{-8}$. For the non-i.i.d. case, we set $\mu = 3.86 \times 10^{-8}$, $\theta = 0.1$ and $\sigma = 8.41 \times 10^{-9}$.
    \item \textbf{Arrival rate parameter.} The typical L2 transaction fee is around $0.09$ USD, equivalent to $3.6 \times 10^{-5}$ ETH at an ETH price of $2500$ USD. Some historical data shows an average of $120$ transactions per L1 block time before EIP-4844. We model the transaction decline by setting a fee cap at $0.27$ USD, or $1.08 \times 10^{-4}$ ETH, above which no transactions occur. This setup defines a linear arrival rate model with parameters $\lambda_0 = 180$ and $k = 1.67 \times 10^6$, based on these two fee thresholds.
    \item \textbf{Single observation.} We set $\kappa = 1$, which means that for each fee update, we only consider the current posting to determine the fee update mechanism.
\end{itemize}

\subsection{Simulation Results}

\noindent \textbf{Decreasing step size and i.i.d. L1 fee case.} The simulation validates the theoretical results for decreasing step sizes with i.i.d. L1 fee, as shown in Appendix \ref{fig-de-iid}. Each graph demonstrates the corresponding parameter converges in the long run. These outcomes support the use of the decreasing step size approach in achieving target efficiency in a stationary environment.

\medskip
\noindent \textbf{Decreasing step size and non-i.i.d. L1 fee case.} In the absence of a formal theorem for almost surely convergence of the sequence $f_t$, $p_t$, $i(t)/t$ and $j(t)/t$ under non-i.i.d. conditions, our simulations serve as a critical empirical test. The results, as shown in Appendix \ref{fig-de-non}, indicate that despite the complexity introduced by the non-i.i.d. nature of L1 fees, all of these parameters still demonstrate a tendency towards convergence over an extended period. This performance suggests that our fee mechanism is robust even in more complicated and realistic L1 gas fee conditions. 

\medskip
\noindent \textbf{Constant step size and non-i.i.d. L1 fee case.} This scenario arguably presents the most practical and realistic conditions for the implementation of our L2 fee mechanism. While we lack a formal convergence theorem for this case, the simulation results, as shown in Appendix \ref{fig-con-non}, are promising, showing that the fees \( f_n \) and \( p_n \) tend to stabilize close to $f^*$ and $p^*$ over time. The histograms indicate that this stabilization occurs within a reasonable neighborhood given the simulation environment. This observation aligns with the theorems presented in Chapter 8 of \cite{harold1997stochastic}, which discusses the convergence behavior for constant step-size cases, indicating asymptotic convergence within a neighborhood of the root. Moreover, the simulation indicates that the relative frequencies \( i(t)/t \) and \( j(t)/t \) of selecting different update mechanisms appear to converge. Although almost surely convergence may not hold in highly stable environments, the robustness and adaptability of our fee mechanism are essential for effectively managing the dynamic and often unstable conditions encountered in real-world applications.

%% file: conclusion.tex
\section{Conclusion}



This paper presents a framework to optimize transaction posting and dynamic fee mechanisms in L2 blockchain systems. By integrating models for transaction arrivals, queue dynamics, and cost structures, we showed that an optimal threshold policy can dictate when to post transactions to L1, balancing operational costs and system performance.

On the pricing side, we developed an L2 fee mechanism that achieves both budget balance and congestion control. Our analysis established the existence and uniqueness of fee levels that satisfy these objectives and examined the convergence properties of the adaptive update mechanism. Simulations confirm that the method efficiently adjusts fees to maintain stability and manage congestion.

Overall, our work bridges transaction posting strategies and fee design, offering a principled method that enhances operational efficiency, financial viability, and congestion management. This integrated perspective paves the way for more resilient and scalable L2 platforms and opens avenues for future research. For example, extending the framework to include blob transactions and other complex settings could further refine our approach to handle diverse transaction types and evolving network demands.

%% file: appendix.tex
\appendix

\section{Proofs} \label{proofs}

\medskip
\noindent\textbf{\sffamily Proof of Theorem \ref{threshold-policy}.}
\begin{proof}
Consider the objective function
\[
J\left( Q_{t},P_{t}\right) =\min_{\left\{ S_{s}\right\} }\mathbb{E}\left[
\sum_{s\geq t}\gamma ^{s-t}\left[ a\left( Q_{s}-S_{s}\right) +\left(
b_{0}+b_{1}S_{s}\right) P_{s}\mathbf{1}{\left\{ S_{s}>0\right\} }\right] \right].
\]%

We first show that when the state is $\left( Q_{t},P_{t}\right) $, we have either $S_{t}^{\ast }=0$ or $S_{t}^{\ast }=Q_{t}$. We then show the following Lemmas.

\begin{lemma} \label{concave-J}
    For any feasible $Q,P$, we have 
    $$2J\left( Q+1,P\right) \geq J\left(Q,P\right) +J\left( Q+2,P\right). $$
\end{lemma}

\begin{proof}

Suppose for any given sequence $\left\{
(A_{s},P_s)\right\} _{s\geq t}$, the optimal strategy for $J\left( Q+1\right) $ is 
$\left\{ S_{s}^{\ast }\right\} _{s\geq t}$, and $u=\min \left\{ \left. s\geq
t\right\vert S_{s}^{\ast }>0\right\} $. We construct the following strategy
for $J\left( Q\right) $ and $J\left( Q+2\right) $. For $J\left( Q\right) $,
we let
\[
S_{s}^{Q}=\left\{ 
\begin{array}{cc}
S_{s}^{\ast }, & \text{when }s\neq u, \\ 
S_{u}^{\ast }-1, & \text{when }s=u.%
\end{array}%
\right. 
\]%
For $J\left( Q+2\right) $, we let
\[
S_{s}^{Q+2}=\left\{ 
\begin{array}{cc}
S_{s}^{\ast }, & \text{when }s\neq u, \\ 
S_{u}^{\ast }+1, & \text{when }s=u.%
\end{array}%
\right. 
\]%
Define the random variable
$$J_\omega (Q,P) = \sum_{s\geq t}\gamma ^{s-t}\left[ a\left( Q_{s}-S_{s}\right) +\left(
b_{0}+b_{1}S_{s}\right) P_{s}\mathbf{1}{\left\{ S_{s}>0\right\} }\right] \Big| (Q_t, P_t) = (Q,P).$$
Then we can show that 
\begin{eqnarray*}
&&2J_\omega\left( \left. Q+1,P\right\vert \left\{ (A_{s},P_s)\right\} _{s\geq t}\right)  \\
&=&2J_\omega\left( \left. Q+1,P\right\vert \left\{ (A_{s},P_s)\right\} _{s\geq t},\left\{
S_{s}^{\ast }\right\} _{s\geq t}\right)  \\
&\geq &J_\omega\left( \left. Q,P\right\vert \left\{ (A_{s},P_s)\right\} _{s\geq t},\left\{
S_{s}^{Q}\right\} _{s\geq t}\right) \\
& & +J_\omega\left( \left. Q+2,P\right\vert \left\{
(A_{s},P_s)\right\} _{s\geq t},\left\{ S_{s}^{Q+2}\right\} _{s\geq t}\right)  \\
&\geq &J_\omega\left( \left. Q,P\right\vert \left\{ (A_{s},P_s)\right\} _{s\geq t}\right)
+J_\omega\left( \left. Q+2,P\right\vert \left\{ (A_{s},P_s)\right\} _{s\geq t}\right) ,
\end{eqnarray*}%
where the first inequality holds because%
\begin{eqnarray*}
&& 2c\left( \left. Q_{s},S_{s}^{\ast },P\right\vert Q_{t}=Q+1\right) \\
&=& c\left(
\left. Q_{s},S_{s}^{Q},P\right\vert Q_{t}=Q\right) +c\left( \left.
Q_{s},S_{s}^{Q+2},P\right\vert Q_{t}=Q+2\right) 
\end{eqnarray*}
for all $s\neq u$ and
\begin{eqnarray*}
&& 2c\left( \left. Q_{u},S_{u}^{\ast },P\right\vert Q_{t}=Q+1\right) \\
&\geq& c\left( \left. Q_{u},S_{u}^{Q},P\right\vert Q_{t}=Q\right) +c\left( \left.
Q_{u},S_{u}^{Q+2},P\right\vert Q_{t}=Q+2\right) .
\end{eqnarray*}
Hence, we have
\begin{align*}
    2J(Q+1,P) &= \mathbb{E}[2J_\omega(Q+1,P) \vert \left\{ (A_{s},P_s)\right\}_{s\geq t}] \\
    &\leq \mathbb{E}[J_\omega(Q,P)\vert \left\{ (A_{s},P_s)\right\}_{s\geq t}] + \mathbb{E}[J_\omega(Q+2,P)\vert \left\{ (A_{s},P_s)\right\}_{s\geq t}] \\
    &= J(Q,P)+J(Q+2,P).
\end{align*}
    
\end{proof}

By Bellman equation, we have
\begin{eqnarray*}
& & J\left( Q_{t},P_{t}\right) \\
&=&\min_{0\leq S_{t}\leq Q_{t}}a\left(
Q_{t}-S_{t}\right) +\left( b_{0}+b_{1}S_{t}\right) P_{t} \mathbf{1}{\left\{
S_{t}>0\right\} }  \\
& &+\gamma \mathbb{E}\left[ J\left(
Q_{t}-S_{t}+A_{t},P_{t+1}\right) |Q_{t}-S_{t},P_{t}\right]  \\
&\triangleq &\min_{0\leq S_{t}\leq Q_{t}}J^{Q}\left( Q_{t},P_{t},S_{t}\right) .
\end{eqnarray*}%

\begin{lemma} \label{Q-function-diffs}
    If $J^{Q}\left( Q_{t},P_{t},S_{t}+1\right)<J^{Q} \left( Q_{t},P_{t},S_{t}\right) $ holds, 
    then we always have $J^{Q}\left( Q_{t},P_{t},S_{t}+2\right) <J^{Q }\left(Q_{t},P_{t},S_{t}+1\right) $. 
\end{lemma}

\begin{proof}
After plugging in the definition of $J^Q$, this Lemma is equivalent to if%
\begin{align*}
&-a + \left( b_{0} 1_{\left\{ S_{t} = 0 \right\}} + b_{1} \right) P_{t} 
-\gamma ( 
\mathbb{E} \big[ J\left( Q_{t} - S_{t} + A_{t}, P_{t+1} \right) \\
&\quad - J\left( Q_{t} - S_{t} - 1 + A_{t}, P_{t+1} \right) 
\big| Q_{t} - S_{t}, P_{t} \big] 
) \leq 0,
\end{align*}
then we have%
\begin{align*}
&-a + b_{1} P_{t} 
- \gamma \left( 
\mathbb{E} \big[ J\left( Q_{t} - S_{t} + A_{t} - 1, P_{t+1} \right) \right. \\
&\quad \left. - J\left( Q_{t} - S_{t} - 2 + A_{t}, P_{t+1} \right) 
\big| Q_{t} - S_{t}, P_{t} \big] \right) \leq 0.
\end{align*}
Comparing these two inequalities, we only need to prove that
\begin{align*}
& -a + b_{1} P_{t} 
- \gamma \left( \mathbb{E} \big[ J\left( Q_{t} - S_{t} + A_{t} - 1, P_{t+1} \right) \right. \\
& \quad \left. - J\left( Q_{t} - S_{t} - 2 + A_{t}, P_{t+1} \right) 
\big| Q_{t} - S_{t}, P_{t} \big] \right) \\
\leq \, & -a + \left( b_{0} 1_{\left\{ S_{t} = 0 \right\}} + b_{1} \right) P_{t} 
- \gamma \left( \mathbb{E} \big[ J\left( Q_{t} - S_{t} + A_{t}, P_{t+1} \right) \right. \\
& \quad \left. - J\left( Q_{t} - S_{t} - 1 + A_{t}, P_{t+1} \right) 
\big| Q_{t} - S_{t}, P_{t} \big] \right),
\end{align*}
which means we only need to show that
\begin{align*}
&\mathbb{E} \big[ 2J\left( Q_{t} - S_{t} + A_{t} - 1, P_{t+1} \right) 
- J\left( Q_{t} - S_{t} + A_{t}, P_{t+1} \right) \\
&\quad - J\left( Q_{t} - S_{t} - 2 + A_{t}, P_{t+1} \right) 
\big| Q_{t} - S_{t}, P_{t} \big] \geq 0.
\end{align*}
This can be directly derived from the Lemma \ref{concave-J}.

\end{proof}

By the Lemma \ref{Q-function-diffs}, we know that if for any feasible $S_t$, we have 
$$J^{Q}\left( Q_{t},P_{t},S_{t}+1\right) \geq J^{Q} \left( Q_{t},P_{t},S_{t}\right),$$
then the optimal policy is $S_t^*=0$. Otherwise, suppose the minimum feasible $S_t$ such that 
$$J^{Q}\left( Q_{t},P_{t},S_{t}+1\right) < J^{Q} \left( Q_{t},P_{t},S_{t}\right)$$ is $S$, which also means that
$$J^{Q} \left( Q_{t},P_{t},0\right) \leq J^{Q} \left( Q_{t},P_{t},1\right) \leq \cdots \leq J^{Q} \left( Q_{t},P_{t},S\right).$$
Then we iteratively use the Lemma to get
$$J^{Q} \left( Q_{t},P_{t},S\right) > J^{Q} \left( Q_{t},P_{t},S+1\right) > \cdots > J^{Q} \left( Q_{t},P_{t},Q_{t}\right).$$
Combining these two inequalities, we can get
$$\min_{0\leq S_{t}\leq Q_{t}}J^{Q}\left( Q_{t},P_{t},S_{t}\right) = \min\{J^{Q} \left( Q_{t},P_{t},0\right), J^{Q} \left( Q_{t},P_{t},Q_{t}\right)\},$$
which means that the optimal value $S_t^*$ is either $0$ or $Q_t$.

We have already proved that the optimal posting strategy $S_t^*$ is either $0$ or $Q_t$. We still need to prove the existence of $Q^*(P_t)$. Then we only need to prove that if $S_t^*(Q_t,P_t) = Q_t$, we always have $S_t^*(Q_t+1,P_t) = Q_t+1$, because in this way, the threshold 
$$Q^*(P_t) = \min \{Q: S_t^*(Q,P_t) = Q\}.$$
The condition of $S_t^*(Q_t,P_t) = Q_t$ means that 
$$J^Q(Q_t,P_t,0) \geq J^Q(Q_t,P_t,Q_t).$$
This is equivalent to
\begin{align*}
    aQ_t+\gamma \mathbb{E}\left[ J\left(Q_{t}+A_{t},P_{t+1}\right) |Q_{t},P_{t}\right] 
    \geq (b_0+b_1Q_t)P_t + \gamma \mathbb{E}\left[ J\left(A_{t},P_{t+1}\right) |P_{t}\right].
\end{align*}
By Lemma \ref{concave-J}, given $Q_t$, we can show that $\forall A_t = A, P_{t+1}$, there is
$$J\left(Q_{t}+A,P_{t+1}\right)-J\left(A,P_{t+1}\right) \leq \frac{Q_t}{Q_t+1} \left(J\left(Q_{t}+1+A,P_{t+1}\right)-J\left(A,P_{t+1}\right)\right),$$
which means that
\begin{eqnarray*}
&&\mathbb{E} \big[ J\left( Q_{t} + A, P_{t+1} \right) - J\left( A, P_{t+1} \right) 
\big| Q_{t} \big] \\
&\leq& \frac{Q_{t}}{Q_{t} + 1} 
\mathbb{E} \big[ J\left( Q_{t} + 1 + A, P_{t+1} \right) - J\left( A, P_{t+1} \right) 
\big| Q_{t} \big].
\end{eqnarray*}
Thus, we can derive
\begin{eqnarray*}
    && \gamma \mathbb{E}\left[J\left(Q_{t}+1+A,P_{t+1}\right)-J\left(A,P_{t+1}\right) | Q_t\right] \\
    &\geq& \frac{Q_t+1}{Q_t} \gamma \mathbb{E}[J\left(Q_{t}+A,P_{t+1}\right)-J\left(A,P_{t+1}\right) | Q_t] \\
    &\geq& \frac{Q_t+1}{Q_t} (-aQ_t + (b_0+b_1Q_t)P_t) \\
    &=& (Q_t+1) \left(-a + \left( \frac{b_0}{Q_t}+b_1 \right)P_t \right) \\
    &\geq& (Q_t+1)\left(-a + \left( \frac{b_0}{Q_t+1}+b_1 \right)P_t \right) \\
    &=& -a(Q_t+1) + (b_0 + b_1(Q_t+1)P_t).
\end{eqnarray*}
By rearranging the inequality and Bellman equation, we can get
$$J^Q(Q_t+1,P_t,0) \geq J^Q(Q_t+1,P_t,Q_t+1),$$
which means that $S_t^*(Q_t+1,P_t) = Q_t+1$. Hence, once $S_t^*(Q_t,P_t) = Q_t$ holds for some queue length $Q_t$, it also holds for larger queue length, which means that
\[
  Q^{*}(P_{t}) = \sup\{Q\geq0:S_t^{*}(Q,P_{t}) = 0\}.
\]
This establishes the existence and uniqueness of the threshold $Q^{*}(P_{t})$. \qed

\end{proof}

\medskip
\noindent\textbf{\sffamily Proof of Property \ref{monotonic-posting}.}
\begin{proof}
Fix two fees \(0<f_1<f_2\le \lambda_0/k\) and denote
\[
\lambda_1 \;=\;\lambda(f_1)=\lambda_0-kf_1,
\qquad
\lambda_2 \;=\;\lambda(f_2)=\lambda_0-kf_2,
\qquad
p \;=\;\frac{\lambda_2}{\lambda_1}\in(0,1).
\]

Let \(\{A_t^{(1)}\}_{t\ge 0}\) be i.i.d.\ Poisson\((\lambda_1)\) random variables.  
For every job arriving in period \(t\) draw an independent Bernoulli\((p)\) mark;  
let \(A_t^{(2)}\) be the number of marked jobs in period \(t\).
Standard thinning shows \(A_t^{(2)}\stackrel{\text{i.i.d.}}{\sim}\operatorname{Pois}(\lambda_2)\) and
\(A_t^{(2)}\le A_t^{(1)}\) almost surely for every \(t\).

Then we consider the L2 system with fee \(f_1\) under its optimal posting policy
\(\{S_t^{(1)}\}_{t\ge 0}\).  
Write \(Q_t^{(1)}\) for the queue length and recall
\[
Q_{t+1}^{(1)} \;=\; Q_t^{(1)}-S_t^{(1)}+A_t^{(1)} .
\]
Within the same probability space, consider the sub system that tracks only the
marked jobs. Its queue length evolves as
\[
\widetilde Q_{t+1} \;=\; \widetilde Q_t - \widetilde S_t + A_t^{(2)},
\]
where $\widetilde S_t$ is the number of marked jobs among $S_t^{(1)}$ posted at time $t$. Since each posted job is marked independently with probability \(p\),
\(\{\widetilde S_t\}\) is feasible for the system fed by arrivals
\(\{A_t^{(2)}\}\), which can also be considered as the system with the fee $f_2$.  Denote the resulting cost stream by
\(
\widetilde c_t := c\bigl(\widetilde S_t,\widetilde Q_t,P_t;f_2\bigr)
\).

For every \(t\), we have \(Q_t^{(1)}-S_t^{(1)} \ge \widetilde Q_t-\widetilde S_t\) and
\(S_t^{(1)} \ge \widetilde S_t\), by the definition of the cost function $c(S,Q,P;f)$, we always have
\[
c\bigl(S_t^{(1)},Q_t^{(1)},P_t;f_1\bigr)
\;\ge \;
c\bigl(\widetilde S_t,\widetilde Q_t,P_t;f_2\bigr).
\]
Since \(p<1\), the probability that there is at least one unmarked job is positive, which means that
\[
\mathbb P\bigl[c\bigl(S_t^{(1)},Q_t^{(1)},P_t;f_1\bigr) - c\bigl(\widetilde S_t,\widetilde Q_t,P_t;f_2\bigr) \ge \min\{a,b_1P_t\}]>0.
\]
This means
\[
\mathbb E\bigl[c(S_t^{(1)},Q_t^{(1)},P_t;f_1)\bigr]
\;>\;
\mathbb E\bigl[c(\widetilde S_t,\widetilde Q_t,P_t;f_2)\bigr]
\quad\text{for every }t.
\]

Let \(S_t^{(2)}\) be the optimal posting policy for fee \(f_2\). Since \(\{\widetilde S_t\}\) is feasible for the \(f_2\) system, optimality yields
\[
\mathbb E[c(S_t^{(2)},Q_t^{(2)},P_t;f_2)]
\;\le\;
\mathbb E[c(\widetilde S_t,\widetilde Q_t,P_t;f_2)].
\]
Thus, we can get
\[
\mathbb E\bigl[c(S_t^{(1)},Q_t^{(1)},P_t;f_1)\bigr]
\;>\;
E[c(\widetilde S_t,\widetilde Q_t,P_t;f_2)]
\;\ge\;
\mathbb E\bigl[c(S_t^{(2)},Q_t^{(2)},P_t;f_2)\bigr].
\]  
This establishes the monotonicity of the expected cost function. \qed
\end{proof}

\medskip
\noindent\textbf{\sffamily Proof of Property \ref{cost-bound}.} 
\begin{proof}
The expected cost with respect to the optimal posting strategy will not be greater than the expected cost with respect to a posting strategy that we post all the transactions at each block. Thus, we have
\begin{align*}
    \mathbb{E}\left[ c(S_t^*,Q_t,P_t;f) \right] &\leq \mathbb{E}\left[ c(Q_t,Q_t,P_t;f) \right] \\
    &= \mathbb{E}\left[ \left(b_{0} + b_{1}Q_{t}\right) P_{t}\mathbf{1}{\left\{ Q_{t}>0\right\} } \right] \\
    &\leq \mathbb{E}\left[ \left(b_{0} + b_{1}Q_{t}\right) P_{t} \right] \\
    &= \left(b_{0} + b_{1} \lambda(f) \right) \mu \\
    &= \left(b_{0} + b_{1} \lambda_0-b_1 kf \right) \mu.
\end{align*}
\qed
\end{proof}

\medskip
\noindent\textbf{\sffamily Proof of Theorem \ref{existence-bb}.} 
\begin{proof}
Suppose $C(f) = \left(b_{0} + b_{1} \lambda_0-b_1 kf \right) \mu$, which is one of the upper bound of $\mathbb{E}\left[ c(S_t^*,Q_t,P_t;f) \right]$ by \textit{Property} \ref{cost-bound}. We first consider the sufficient condition that there is at least one root for $\lambda(f) f = C(f)$. Since 
$$C(\lambda_0/k+b_0/(b_1k))=0$$ and $C(f)$ is monotonically decreasing, the sufficient condition that $\lambda(f) f = C(f)$ has at least one root can be 
$$\max \left\{\lambda(f)f\right\} \geq C(0).$$
This is equivalent to 
$$\frac{\lambda_0^2}{4k} \geq (b_0+b_1 \lambda_0) \mu.$$

If $\lambda_0^2/(4k) \geq (b_0+b_1 \lambda_0) \mu,$ we know that there exists $\hat{f} \in [0,\lambda_0/(2k)]$ such that $\lambda(\hat{f})\hat{f} = C(\hat{f})$. When $f = 0$, we have $$\lambda(f)f=0<\mathbb{E}\left[ c(S_t^*,Q_t,P_t;0) \right].$$
When $f=\hat{f}$, we have
$$\lambda(\hat{f})\hat{f} = C(\hat{f}) \geq \mathbb{E}\left[ c(S_t^*,Q_t,P_t;\hat{f})\right].$$
Thus, there exists a root $f^* \in [0,\hat{f}]$, such that
$$\lambda(f^*)f^* = \mathbb{E}\left[ c(S_t^*,Q_t,P_t;f^*)\right].$$

The uniqueness of the root over the interval \(f \in \left[0, \lambda_0/(2k)\right]\) follows from the fact that the left-hand side is strictly monotonically increasing in \(f\). This is because \(\mathbb{E}[A(s; f)] = \lambda_0 - kf\) is linear and decreasing, making \(\mathbb{E}[A(s; f) f] = f(\lambda_0 - kf)\) strictly increasing on \([0, \lambda_0/(2k)]\), while \(\mathbb{E}[c(S_s^*, Q_s, P_s; f)]\) is strictly decreasing in \(f\) by Property~\ref{monotonic-posting}. \qed

\end{proof}

\medskip
\noindent\textbf{\sffamily Proof of Proposition \ref{renewal-thm}.} 

\begin{proof}
We design $A_s^{(i)}(g) = A(l;g)$ and $c^{(i)}(S_s^*,Q_s,P_s;g) = c(S_l^*,Q_l,P_l;g)$, where 
$$l = \sum_{j=1}^{i-1} \tau_j(g) + s.$$
By ergodic theorem and the strong law of large number, we have 
\begin{eqnarray*}
    \mathbb{E}[X(t;g)] &=& \lim_{N\rightarrow \infty} \frac{1}{N} \sum_{i=1}^N \left[ \sum_{s=1}^{\tau_i(g)} \left (A_s^{(i)}(g) g - c^{(i)}(S_s^*,Q_s,P_s;g) \right)  \right]\\
    &=& \lim_{N\rightarrow \infty} \frac{\sum_{i=1}^N \tau_i(g) }{N} \frac{1}{\sum_{i=1}^N \tau_i(g)} \sum_{i=1}^N \sum_{s=1}^{\tau_i(g)} \left( A_s^{(i)}(g) g - c^{(i)}(S_s^*,Q_s,P_s;g) \right)  \\
    &=& \lim_{N\rightarrow \infty} \frac{\sum_{i=1}^N \tau_i(g) }{N} \\
    &&\cdot \lim_{N\rightarrow \infty} \frac{1}{\sum_{i=1}^N \tau_i(g)} \sum_{i=1}^N \sum_{s=1}^{\tau_i(g)} \left( A_s^{(i)}(g) g - c^{(i)}(S_s^*,Q_s,P_s;g) \right)  \\
    &=& \mathbb{E}[\tau(g)]\mathbb{E}[A(g) g -c(S_s^*,Q_s,P_s;g)].
\end{eqnarray*}
By the definition of $f^*$, we will have $\mathbb{E}[X(s;f^*)]=0$ after plugging in $g=f^*$.

Similarly, we can also get the result for $\mathbb{E}[Y(t;g)]$. \qed
\end{proof}

\medskip
\noindent\textbf{\sffamily Proof of Theorem \ref{fee-convergence-thm}.} 
\begin{proof}
First, we show that $i(t) \rightarrow \infty$ and $j(t) \rightarrow \infty$ almost surely. If $i(t) < \infty$, there must $\exists M \geq 0$, such that $\forall t>M$, we have $\delta_t = 0$.
Consider
\begin{align*}
    P(i(t)<\infty) &\leq \sum_{l=0}^{\infty} P(M=l) P(\forall t>M, \delta_t = 0 | M=l) \\
    &\leq \sum_{l=0}^{\infty}  P(\forall t>M, \delta_t = 0 | M=l) \\
    &= \sum_{l=0}^{\infty}  P(\forall t>l, \delta_t = 0 | M=l).
\end{align*}
For all $l \geq 0$, we have
\begin{eqnarray*}
    && P(\forall t>l, \delta_t = 0 | M=l) \\
    &=& \prod_{t=l+1}^\infty P(\delta_t = 0 | M=l, \delta_{t-1} = \cdots = \delta_{l+1} = 0) \\
    &\leq& \prod_{t=l+2}^\infty P(\delta_t = 0 | M=l, \delta_{t-1} = \cdots = \delta_{l+1} = 0) \\
    &\leq& \prod_{t=l+2}^\infty \max_{p\in [0, \lambda_0/(2k)]} P(\delta_t = 0 | M=l, \delta_{t-1} = \cdots = \delta_{l+1} = 0, p_{j(t-1)}=p) \\
    &=& \prod_{t=l+2}^\infty \max_{p\in [0, \lambda_0/(2k)]} P(\delta_t = 0 | \delta_{t-1} = 0, p_{j(t-1)}=p) \\
    &=& \prod_{t=l+2}^\infty \max_{p\in [0, \lambda_0/(2k)]} P\left(\sum_{s=\kappa (t-1)+1}^{s=\kappa t} X(s; p) \geq 0\right) \\
    &\leq& \prod_{t=l+2}^\infty (1-\rho) \\
    &=& 0,
\end{eqnarray*}
where 
$$\rho = 1- \max_{p\in [0, \lambda_0/(2k)]} P\left(\sum_{s=\kappa (t-1)+1}^{s=\kappa t} X(s; p) \geq 0\right),$$
and it is easy to show that $\rho > 0$ by using the compact property of interval $[0,\lambda_0/(2k)]$. Hence, we have
\begin{align*}
    P(i(t)<\infty) &\leq \sum_{l=0}^{\infty}  P(\forall t>l, \delta_t = 0 | M=l) \\
    &= \sum_{l=0}^{\infty} 0 \\
    &= 0,
\end{align*}
which means that $P(i(t)=\infty) = 1$. Similarly, we have $P(j(t)=\infty) = 1$. Thus, $i(t) \rightarrow \infty$ and $j(t) \rightarrow \infty$ almost surely.

\medskip
Second, we show that $f_t \rightarrow f^*$ and $p_t \rightarrow p^*$ almost surely. Define the update stopping times  
\[
T_n^{f}:=\inf\{t\ge0:\delta_t=1,\;i(t)=n\},
\]
and set \(f_n:=g_{T_n^{f}}\). For \(n\ge0\), suppose
\[
\widehat X_n:=\frac1\kappa\sum_{s=1}^{\kappa}X\bigl(T_n^{f}+s;f_n\bigr).
\]
Since the queue empties when a budget‑balance update is triggered, by the strong Markov property
\[
\mathbb E\!\bigl[
  \varphi\!\bigl(X(T_n^{f}+1;f_n),\cdots,X(T_n^{f}+\kappa;f_n)\bigr)
  \,\big|\,\mathcal F_{T_n^{f}}
\bigr]
=\mathbb E\!\bigl[
     \varphi\!\bigl(X(1;f_n),\dots,X(\kappa;f_n)\bigr)
  \bigr]\quad\text{a.s.},
\]
for any bounded measurable \(\varphi:\mathbb R^{\kappa}\!\to\!\mathbb R\).
With step sizes \(a_n:=a/(n+1)\), we have the recursion
\[
f_{n+1}
=\Pi_{[0,\lambda_0/(2k)]}\bigl(f_n-a_n\widehat X_{n+1}\bigr).
\]
Since we have already proved that $i(t) \rightarrow \infty$ almost surely, by the strong Markov property and the recursion, we can derive $f_n \rightarrow f^*$ almost surely straightforwardly by using the Theorem 2.1 at Page 127 of \cite{harold1997stochastic} with Theorem \ref{existence-bb} and Theorem \ref{existence-cc}.
Similarly, we can also show that $p_t \rightarrow p^*$ almost surely.

\medskip
Third, we show that the long-run average proportions $i(t)/t \rightarrow \pi_f$ and $j(t)/t  \rightarrow \pi_p$ almost surely. We have already established that the update rules for the fees satisfy
\[
f_t\longrightarrow f^{*},\qquad 
p_t\longrightarrow p^{*},\qquad\text{a.s.}
\]
and that, when the fees are frozen at $(f^{*},p^{*})$, the indicator process $\delta_t\in\{0,1\}$ evolves as an irreducible, aperiodic two state Markov chain with transition matrix
\[
P=\begin{pmatrix}
P_{00} & P_{01}\\
P_{10} & P_{11}
\end{pmatrix},\qquad 
P_{uv}\in(0,1),\;u,v\in\{0,1\}.
\]
Then we show that $\forall \epsilon>0$, there exists \(\alpha>0\) such that for any matrix \(P'=(P'_{uv})\) satisfying
\[
\max_{u,v\in\{0,1\}} \bigl|P_{uv}-P'_{uv}\bigr| < \alpha,
\]
the corresponding stationary distribution \(\pi(P')=(\pi_f',\pi_p')\) obeys
\[
\bigl|\pi_f - \pi_f'\bigr| < \epsilon
\quad\text{and}\quad
\bigl|\pi_p - \pi_p'\bigr| < \epsilon.
\]
This follows from the explicit formulas
\[
\pi_f = \frac{P_{10}}{P_{10}+P_{01}}
\quad\text{and}\quad
\pi_p = \frac{P_{01}}{P_{10}+P_{01}},
\]
which are uniformly continuous functions. Hence, there exists \(\alpha>0\), such that for all \(\max_{u,v}|P_{uv}-P'_{uv}|<\alpha\), we have $\bigl|\pi_f - \pi_f'\bigr| < \epsilon$ and $\bigl|\pi_p - \pi_p'\bigr|  = \bigl|\pi_f - \pi_f'\bigr| < \epsilon$.  
Since $f_t \rightarrow f^*$ and $p_t \rightarrow p^*$ almost surely, and $i(t),j(t) \rightarrow \infty$ almost surely, there must exist a random variable $\tau<\infty$ such that for all $t>\tau$:
\begin{enumerate}
    \item if $\delta_t=1$, $P(\delta_{t+1}=1 | \delta_t=1, \mathcal{F}_t)=P(\delta_{t+1}=1 | \delta_t=1, g_t)\geq P_{11}-\alpha$; 
    \item if $\delta_t=0$, $P(\delta_{t+1}=1 | \delta_t=0, \mathcal{F}_t)=P(\delta_{t+1}=1 | \delta_t=0, g_t)\geq P_{01}-\alpha$
\end{enumerate}
with probability $1$. This is because the sequence $\{f_t\}_{t>\tau}$ and $\{p_t\}_{t>\tau}$ will fall into very close neighborhoods of $f^*$ and $p^*$ respectively, by the definition of almost surely convergence. Let's define
\[
T_l=\inf\{t: i(t)=\ell\},
\quad \text{and} \quad
t_l^{(f)}=T_{l+1}-T_l.
\]
Hence, in order to show that $i(t)/t \rightarrow \pi_f$, we only need to show that
$$\frac{\sum_{l=1}^n t_l^{(f)}}{n} \rightarrow \frac{1}{\pi_f}.$$
We introduce a double array of i.i.d.\ random variables 
\[
\left\{ \xi_{l,s} \sim \mathrm{Unif}[0,1] \right\}_{l,s \ge 1},
\]
where each \(\xi_{l,s}\) is independently drawn from the uniform distribution on \([0,1]\). 
For some time \( t \) such that \(\delta_t = 1\), suppose that \( l = i(t) \). The subsequent evolution of the chain is determined by the entries in the \( l \)-th row of the double array \(\{\xi_{l,s}\}_{s \geq 1}\):
\begin{enumerate}
    \item At step \( t + 1 \), we set
    \[
    \delta_{t+1} =
    \begin{cases}
        1, & \text{if } \xi_{l,1} \leq P(\delta_{t+1}=1 | \delta_t=1, g_t), \\[6pt]
        0, & \text{otherwise}.
    \end{cases}
    \]

    \item If \(\delta_{t+1}=1\), the procedure stops for this row, and we move immediately to the next row (\(\{\xi_{l+1,s}\}_{s \geq 1}\)).

    \item If \(\delta_{t+1}=0\), we remain on row \( l \). At step \( t + 2 \), we then use the next entry in this row, setting
    \[
    \delta_{t+2} =
    \begin{cases}
        1, & \text{if } \xi_{l,2} \leq P(\delta_{t+2}=1 | \delta_{t+1}=0, g_{t+1}), \\
        0, & \text{otherwise}.
    \end{cases}
    \]

    \item More generally, for each \( s \geq 2 \), if we have not yet returned to state \( 1 \) by step \( t + s - 1 \), we continue using the subsequent entries of row \( l \) as follows:
    \[
    \delta_{t+s} =
    \begin{cases}
        1, & \text{if } \xi_{l,s} \leq P(\delta_{t+s}=1 | \delta_{t+s-1}=0, g_{t+s-1}), \\
        0, & \text{otherwise}.
    \end{cases}
    \]

    \item This process continues until the chain returns to the state \( 1 \). After returning to state \( 1 \), we increment the row index from \( l \) to \( l+1 \) and repeat the entire process with the new row.
\end{enumerate}
Thus, the random time spent on row \( l \) is precisely the duration until the chain first returns to state \( 1 \), formally defined by
\[
t_l^{(f)} = \inf\{s \geq 1 : \delta_{T_l + s} = 1\}, \quad \text{where} \quad T_l = \inf\{t : i(t) = l\}.
\]
Now, let's consider
\[
\hat P_{11}=P_{11}-\alpha,\quad
\hat P_{01}=P_{01}-\alpha,\quad
\hat P_{10}=1-\hat P_{11},\quad
\hat P_{00}=1-\hat P_{01},
\]
that defines a valid transition matrix \(\hat P\). By construction, for all $t>\tau$, we have the component-wise inequalities
$$
P(\delta_{t+1}=1 | \delta_t=1, g_t) \geq \hat{P}_{11}, \quad P(\delta_{t+1}=1 | \delta_t=0, g_t)\geq \hat{P}_{01}.
$$
We introduce another auxiliary sequence $\{\beta_t\}_{t\geq 0}$. We define the evolution of $\{\beta_t\}$ using successive rows of the array $\{\xi_{l,s}\}_{s\geq1}$. Specifically, we set $\beta_0=1$ initially and suppose we are currently at row $l$. Then:
\begin{enumerate}
\item At step $t+1$, set
$$
\beta_{t+1} =
    \begin{cases}
        1, & \text{if } \xi_{l,1}\leq \hat{P}_{11},\\[6pt]
        0, & \text{otherwise}.
    \end{cases}
$$
\item If \(\beta_{t+1}=1\), stop using row \(l\) and increment the row index to \(l+1\). If \(\beta_{t+1}=0\), continue using the subsequent entry \(\xi_{l,2}\) at step \(t+2\).

\item For each subsequent step \(t+s\) with \(s\geq2\), if \(\beta_{t+s-1}=0\), define
\[
\beta_{t+s}=
\begin{cases}
    1, & \text{if } \xi_{l,s}\leq \hat{P}_{01},\\[6pt]
    0, & \text{otherwise}.
\end{cases}
\]

\item Continue this process until \(\beta\) returns to state \(1\). After returning to state \(1\), increment the row index to \(l+1\) and repeat the procedure.
\end{enumerate}
We also define the return time as
\[
\hat{t}_l = \inf\{s\geq 1:\beta_{t+s}=1\}.
\]
Since $\forall t>\tau$,
$$
P(\delta_{t+1}=1 | \delta_t=1, g_t) \geq \hat{P}_{11}, \quad P(\delta_{t+1}=1 | \delta_t=0, g_t)\geq \hat{P}_{01},
$$
we have the inequality of the corresponding return times:
$$
\hat{t}_l \leq t^{(f)}_l,\quad \text{for all } l>i(\tau).
$$
The sequence $\{\hat{t}_l\}$ corresponds exactly to the return times of a stationary Markov chain with transition matrix $\hat{P}$. Thus, by the ergodic theorem for irreducible and aperiodic Markov chains,
$$
\frac{\sum_{l=1}^n\hat{t}_l}{n}\rightarrow \frac{1}{\hat{\pi}_f},
$$
almost surely.
\(\max_{u,v}|P_{uv}-\hat P_{uv}|<\alpha\), we have $\bigl|\pi_f - \hat \pi_f\bigr| < \epsilon$.
Thus,
$$
\liminf_{n\rightarrow\infty}\frac{\sum_{l=1}^nt^{(f)}_l}{n}
= \liminf_{n\rightarrow\infty}\frac{\sum_{l=i(\tau)+1}^nt^{(f)}_l}{n}
\geq\lim_{n\rightarrow\infty}\frac{\sum_{l=i(\tau)+1}^n\hat{t}_l}{n}>\frac{1}{\pi_f+\epsilon}.
$$
Similarly, we can also get an upper bound:
$$
\limsup_{n\rightarrow\infty}\frac{\sum_{l=1}^nt^{(f)}_l}{n} < \frac{1}{\pi_f-\epsilon}.
$$
Since $\epsilon>0$ is arbitrary, letting $\epsilon\rightarrow0$, we conclude
$$
\frac{\sum_{l=1}^nt^{(f)}_l}{n}\rightarrow \frac{1}{\pi_f},
$$
almost surely, which means that $i(t)/t \rightarrow \pi_f$ almost surely. Similarly, we can also get $j(t)/t \rightarrow \pi_p$ almost surely.

\medskip
Fourth, we show the convergence property of the stationary distribution $(\pi_f,\pi_p)$ with respect to the transition matrix $P$ as $\kappa \rightarrow \infty$. Without loss of generality, we assume \(f^{*}>p^{*}\). Define
\[
X^{\kappa}_t = \sum_{s=\kappa t+1}^{\kappa(t+1)}\bigl(A(s;p^{*})\,p^{*}-c(S_s^{*},Q_s,P_s;p^{*})\bigr)
\]
and
\[
Y^{\kappa}_t = \sum_{s=\kappa t+1}^{\kappa(t+1)}(\bar{\lambda}-A(s;f^{*})).
\]
Suppose
\[
\mu_X=\mathbb E[A(s;p^{*})\,p^{*}-c(S_s^{*},Q_s,P_s;p^{*})],
\quad
\mu_Y=\mathbb E[\bar\lambda-A(s;f^{*})]
\]
and
\[
\sigma_X^{2}=\operatorname{Var}[A(s;p^{*})\,p^{*}-c(S_s^{*},Q_s,P_s;p^{*})],\quad
\sigma_Y^{2}=\operatorname{Var}[\bar\lambda-A(s;f^{*})].
\]
Since \(f^{*}>p^{*}\), we have \(\lambda(f^{*})<\bar\lambda\), which means \(\mu_Y>0\).
Similarly, we have \(\mu_X<0\).
By the Central Limit Theorem, the transition probabilities of the chain for a given \(\kappa\) satisfy
$$P_{10} = P\left(Y^\kappa_n < 0\right) \rightarrow 0. $$
Similarly, we also have $P_{00} \rightarrow 0.$ Thus, the stationary distribution
\[
(\pi_f,\pi_p)\rightarrow(1,0)
\]
as $\kappa\rightarrow\infty$.
Next, we characterize the rate of convergence using the Berry–Esseen Theorem. Define 
\[
T_Y = \frac{Y_{\kappa}-\kappa\mu_Y}{\sigma_Y\sqrt{\kappa}}.
\]
By Berry–Esseen Theorem, there exists a universal constant \(C>0\) such that for all \(y\),
\[
|P(T_Y\le y)-\Phi(y)|\le\frac{C\rho_Y}{\sigma_Y^3\sqrt{\kappa}},
\]
where
\[\rho_Y=\mathbb{E}[|\bar\lambda - A(s;f^*)-\mu_Y|^3].\] 
Since \(\mu_Y>0\), we set \(y=-\mu_Y\sqrt{\kappa}/\sigma_Y\) and obtain
\[
P(Y_{\kappa}<0)=P(T_Y\leq y)\leq \frac{C_Y}{\sqrt{\kappa}},
\]
for some constant \(C_Y>0\). Similarly, we also have
\[
P(X_{\kappa}\geq0)\leq\frac{C_X}{\sqrt{\kappa}},
\]
for some constant \(C_X>0\). Thus, we have
\[
P_{10}=O(\kappa^{-1/2})
\quad \text{and} \quad
P_{00}=O(\kappa^{-1/2}).
\]
Expressing the stationary distribution explicitly as
\[
\pi_f = \frac{P_{01}}{P_{01}+P_{10}},
\]
we can conclude taht the convergence rate of the stationary distribution is
\[
|\pi_f - 1|=|\pi_p - 0|=O(\kappa^{-1/2}).
\]
Similarly, the case \(f^{*}<p^{*}\) follows from a symmetric argument, which gives
\[
(\pi_f,\pi_p)\rightarrow(0,1),\quad |\pi_f - 0|=|\pi_p - 1|=O(\kappa^{-1/2}).
\] \qed

\end{proof}

\newpage

\section{Figures}
\subsection{Decreasing step size and i.i.d. L1 fee case} \label{fig-de-iid}

\begin{figure}[H]
    \centering
    \begin{minipage}{0.49\textwidth}
        \centering
        \includegraphics[width=\textwidth]{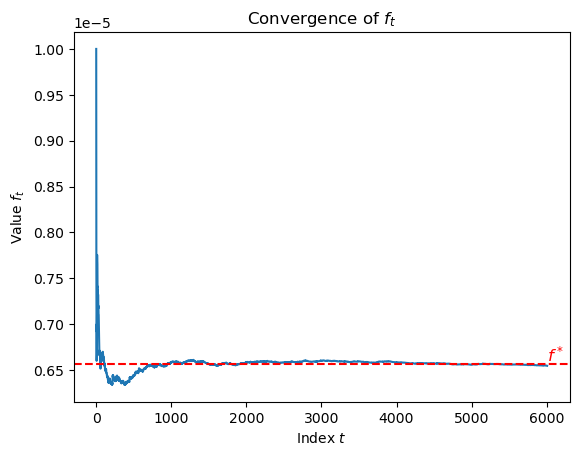}
    \end{minipage}\hfill
    \begin{minipage}{0.49\textwidth}
        \centering
        \includegraphics[width=\textwidth]{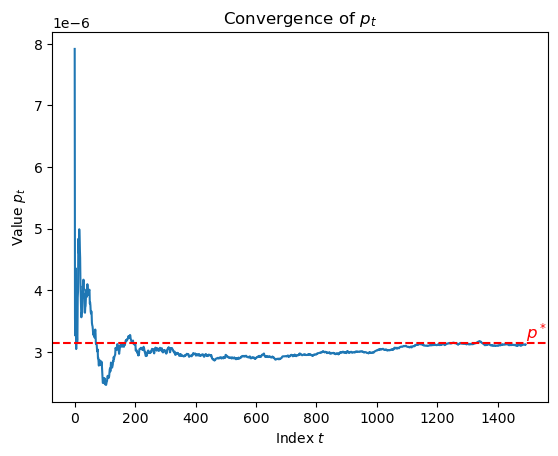}
    \end{minipage}

    \begin{minipage}{0.49\textwidth}
        \centering
        \includegraphics[width=\textwidth]{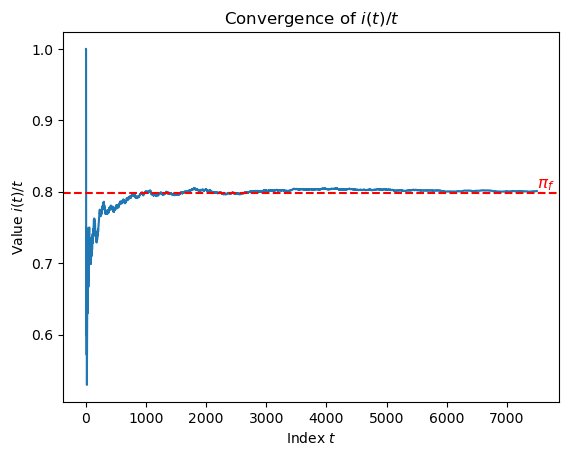}
    \end{minipage}\hfill
    \begin{minipage}{0.49\textwidth}
        \centering
        \includegraphics[width=\textwidth]{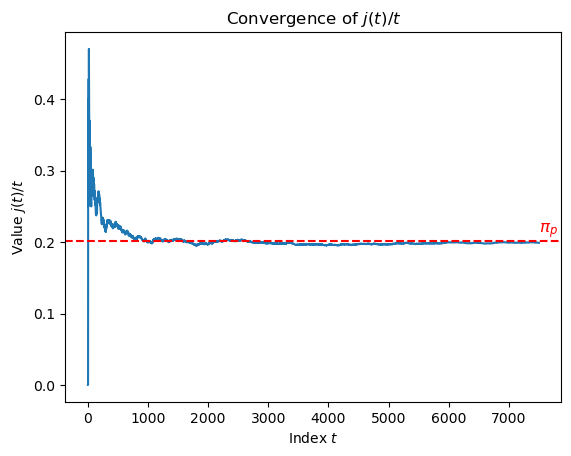}
    \end{minipage}
    \caption{
    Illustration of the fee update mechanism under decreasing step size and 
    i.i.d. L1 fees. The top row displays the evolution of the posting fees over time, with red dashed lines marking the theoretical optimal levels \(f^*\) and \(p^*\). In the bottom row, the plots show the proportion of times each type of fee update mechanism (budget balance vs.\ congestion control) is selected. All four trajectories exhibit convergence behavior in line with the theoretical results.}
\end{figure}

\newpage

\subsection{Decreasing step size and non-i.i.d. L1 fee case} \label{fig-de-non}

\begin{figure}[H]
    \centering
    \begin{minipage}{0.49\textwidth}
        \centering
        \includegraphics[width=\textwidth]{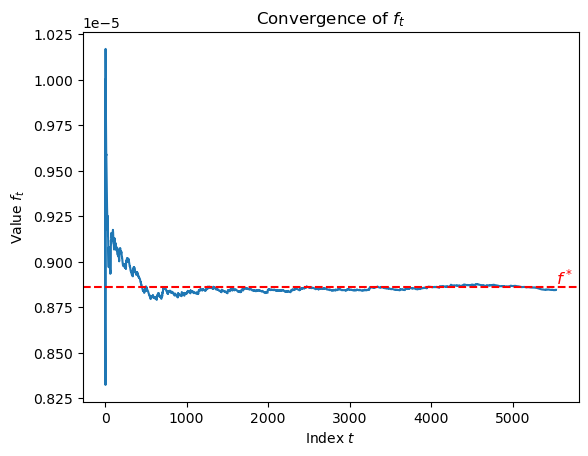}
    \end{minipage}\hfill
    \begin{minipage}{0.49\textwidth}
        \centering
        \includegraphics[width=\textwidth]{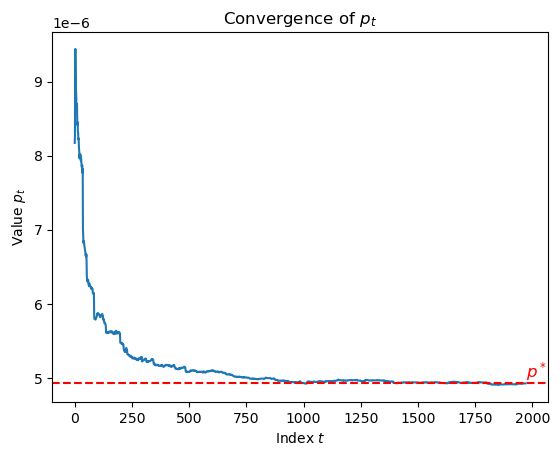}
    \end{minipage}
    
    \begin{minipage}{0.49\textwidth}
        \centering
        \includegraphics[width=\textwidth]{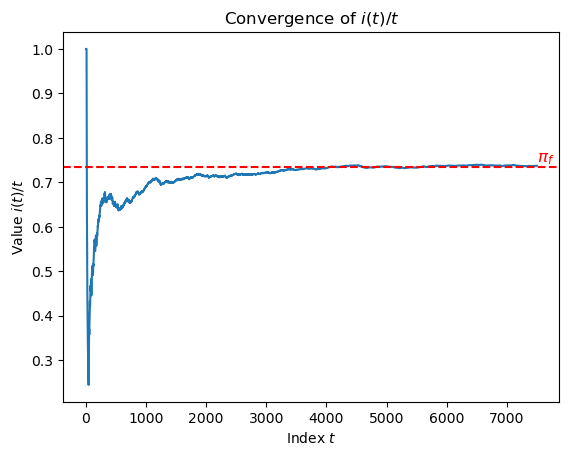}
    \end{minipage}\hfill
    \begin{minipage}{0.49\textwidth}
        \centering
        \includegraphics[width=\textwidth]{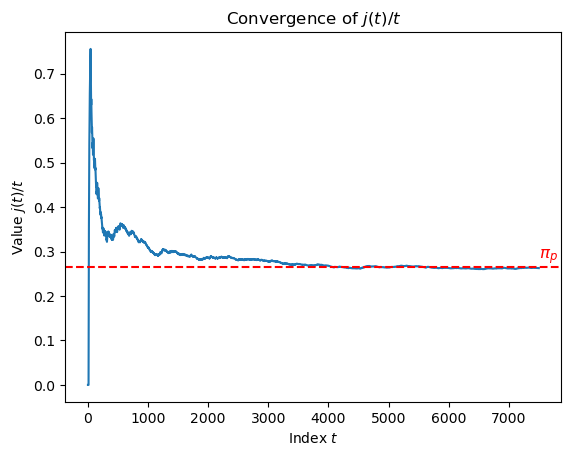}
    \end{minipage}
    \caption{%
    Illustration of the fee update mechanism under decreasing step size and 
    non-i.i.d. L1 fees. The top row displays the evolution of the posting fees over time, with red dashed lines marking the theoretical optimal levels \(f^*\) and \(p^*\). In the bottom row, the plots show the proportion of times each type of fee update mechanism (budget balance vs.\ congestion control) is selected. All four trajectories exhibit convergence behavior in line with the theoretical results.}

\end{figure}

\newpage

\subsection{Constant step size and non-i.i.d. L1 fee case.} \label{fig-con-non}

\begin{figure}[H]
    \centering
    \begin{minipage}{0.49\textwidth}
        \centering
        \includegraphics[width=\textwidth]{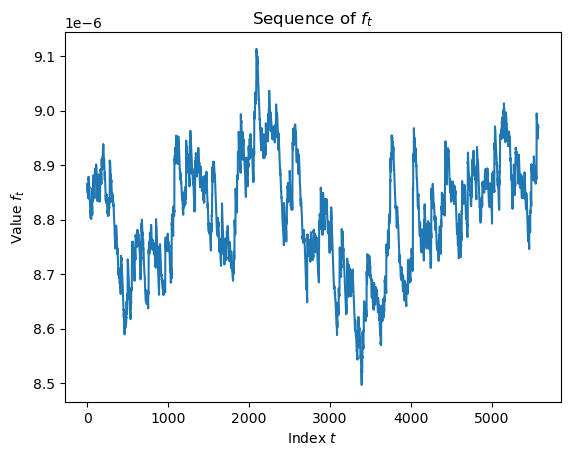}
    \end{minipage}\hfill
    \begin{minipage}{0.49\textwidth}
        \centering
        \includegraphics[width=\textwidth]{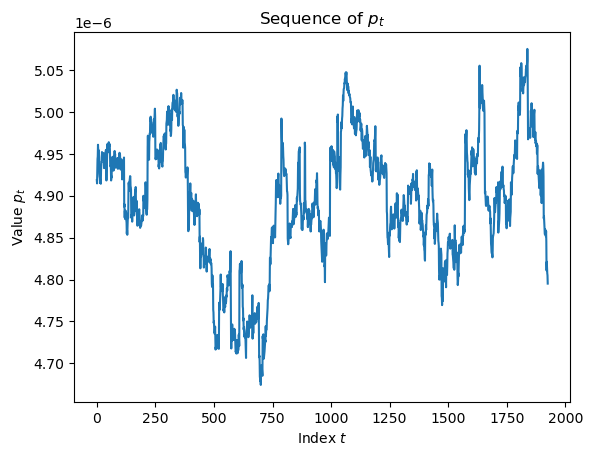}
    \end{minipage}

    \begin{minipage}{0.49\textwidth}
        \centering
        \includegraphics[width=\textwidth]{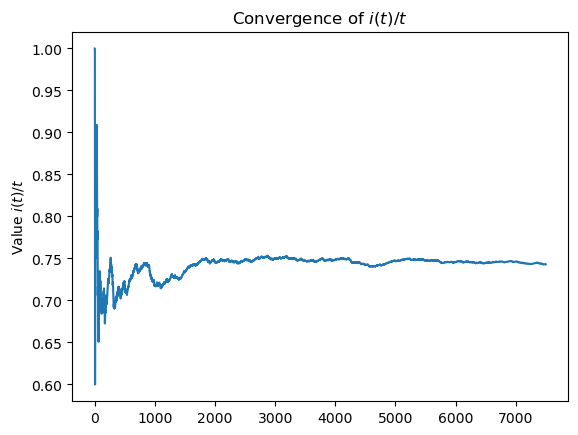}
    \end{minipage}\hfill
    \begin{minipage}{0.49\textwidth}
        \centering
        \includegraphics[width=\textwidth]{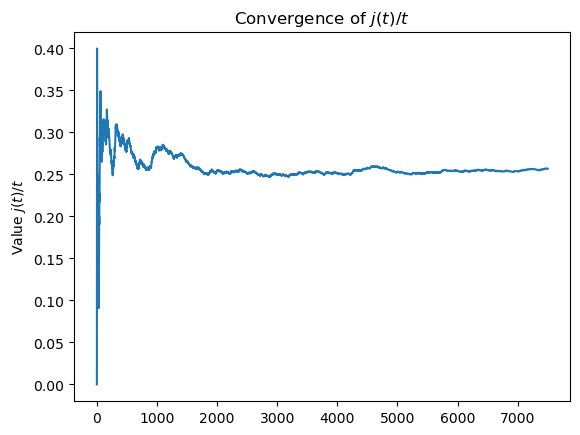}
    \end{minipage}

    \begin{minipage}{0.49\textwidth}
        \centering
        \includegraphics[width=\textwidth]{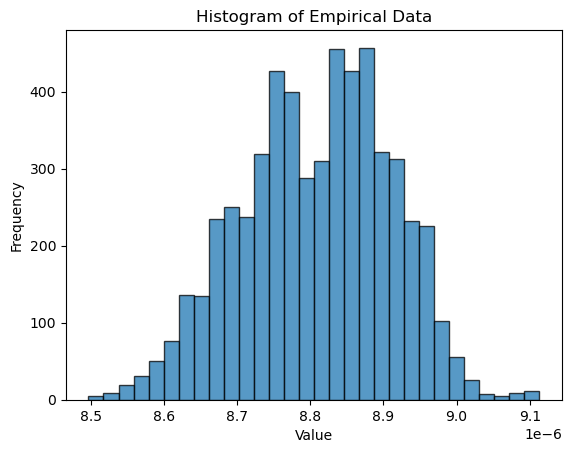}
    \end{minipage}\hfill
    \begin{minipage}{0.49\textwidth}
        \centering
        \includegraphics[width=\textwidth]{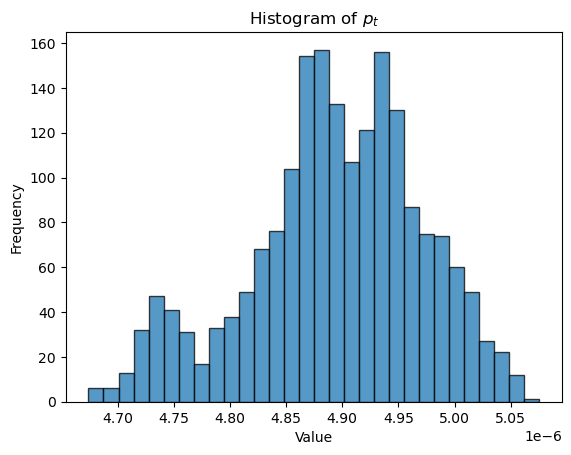}
    \end{minipage}
    \caption{%
    Illustration of the fee update mechanism under constant step size and non‐i.i.d. L1 fees. The top row shows time‐series of the fee updates \(f_t\) and \(p_t\) under a mean‐reverting L1 fee process. In the middle row, the plots track the long‐run fraction of times each update rule is chosen, indicating how often the budget‐balance versus congestion‐control mechanism is active. The bottom row presents empirical histograms of \(f_t\) and \(p_t\). Although no formal convergence theorem applies in this setting, the simulation demonstrates that both fees remain near the targets, and that the selection frequencies also stabilize, highlighting the mechanism's robustness under realistic market dynamics.}
\end{figure}

%% file: main.bbl
\begin{thebibliography}{10}

\bibitem{bar2023optimal}
Yogev Bar-On and Yishay Mansour.
\newblock Optimal publishing strategies on a base layer.
\newblock {\em arXiv preprint arXiv:2312.06448}, 2023.

\bibitem{basu2019towards}
Soumya Basu, David Easley, Maureen O'Hara, and Emin~G{\"u}n Sirer.
\newblock Towards a functional fee market for cryptocurrencies.
\newblock {\em arXiv preprint arXiv:1901.06830}, 2019.

\bibitem{bousfield2022arbitrum}
Lee Bousfield, Rachel Bousfield, Chris Buckland, Ben Burgess, Joshua Colvin, Edward~W Felten, Steven Goldfeder, Daniel Goldman, Braden Huddleston, H~Kalonder, et~al.
\newblock Arbitrum nitro: A second-generation optimistic rollup, 2022.

\bibitem{crapis2023eip4844}
Davide Crapis.
\newblock Eip-4844 fee market analysis.
\newblock \url{{https://ethresear.ch/t/eip-4844-fee-market-analysis/15078}}, March 2023.

\bibitem{crapis2023eip}
Davide Crapis, Edward~W Felten, and Akaki Mamageishvili.
\newblock Eip-4844 economics and rollup strategies.
\newblock {\em arXiv preprint arXiv:2310.01155}, 2023.

\bibitem{multidimensional}
Davide Crapis, Ciamac~C. Moallemi, and Shouqiao Wang.
\newblock Optimal dynamic fees for blockchain resources.
\newblock {\em CoRR}, abs/2309.12735, 2023.

\bibitem{harold1997stochastic}
J~Harold, G~Kushner, and George Yin.
\newblock Stochastic approximation and recursive algorithm and applications.
\newblock {\em Application of Mathematics}, 35(10), 1997.

\bibitem{katti1968handbook}
SK~Katti and A~Vijaya Rao.
\newblock Handbook of the poisson distribution, 1968.

\bibitem{lavi2022redesigning}
Ron Lavi, Or~Sattath, and Aviv Zohar.
\newblock Redesigning bitcoin’s fee market.
\newblock {\em ACM Transactions on Economics and Computation}, 10(1):1--31, 2022.

\bibitem{opt_chaos}
Stefanos Leonardos, Dani{\"{e}}l Reijsbergen, Barnab{\'{e}} Monnot, and Georgios Piliouras.
\newblock Optimality despite chaos in fee markets.
\newblock {\em CoRR}, abs/2212.07175, 2022.

\bibitem{little2008little}
John~DC Little and Stephen~C Graves.
\newblock Little's law.
\newblock {\em Building intuition: insights from basic operations management models and principles}, pages 81--100, 2008.

\bibitem{eff_batch_posting}
Akaki Mamageishvili and Edward~W. Felten.
\newblock Efficient {Rollup} batch posting strategy on {Base Layer}.
\newblock {\em CoRR}, abs/2212.10337, 2022.
\newblock \href {https://arxiv.org/abs/2212.10337} {\path{arXiv:2212.10337}}, \href {https://doi.org/10.48550/arXiv.2212.10337} {\path{doi:10.48550/arXiv.2212.10337}}.

\bibitem{meister2024gas}
Bernhard~K Meister and Henry~CW Price.
\newblock Gas fees on the ethereum blockchain: from foundations to derivative valuations.
\newblock {\em Frontiers in Blockchain}, 7:1462666, 2024.

\bibitem{puterman2014markov}
Martin~L Puterman.
\newblock {\em Markov decision processes: discrete stochastic dynamic programming}.
\newblock John Wiley \& Sons, 2014.

\bibitem{roughgarden2020transaction}
Tim Roughgarden.
\newblock Transaction fee mechanism design for the ethereum blockchain: An economic analysis of eip-1559.
\newblock {\em arXiv preprint arXiv:2012.00854}, 2020.

\bibitem{pantera2024pricing}
Matt Stephenson and Ally Zach.
\newblock Getting the pricing right in crypto, 2024.
\newblock URL: \url{https://panteracapital.com/research-getting-the-pricing-right-in-crypto/}.

\bibitem{yaish2023correct}
Aviv Yaish and Aviv Zohar.
\newblock Correct cryptocurrency asic pricing: Are miners overpaying?
\newblock In {\em 5th Conference on Advances in Financial Technologies (AFT 2023)}. Schloss Dagstuhl-Leibniz-Zentrum f{\"u}r Informatik, 2023.

\bibitem{yao2018incentive}
Andrew Chi-Chih Yao.
\newblock An incentive analysis of some bitcoin fee designs.
\newblock {\em arXiv preprint arXiv:1811.02351}, 2018.

\end{thebibliography}
